\documentclass[a4paper,UKenglish,cleveref, autoref, numberwithinsect]{lipics-v2021}

\pdfoutput=1 
\hideLIPIcs  

\title{Ideal Membership in Polynomial Calculus: Complexity and Reductions} 

\titlerunning{Ideal Membership in Polynomial Calculus} 

\author{Alex {Bortolotti}}{University of Applied Sciences and Arts of Southern Switzerland, IDSIA, Switzerland}{alex.bortolotti@supsi.ch}{https://orcid.org/0009-0008-0988-4873}{}

\author{Monaldo {Mastrolilli}}{University of Applied Sciences and Arts of Southern Switzerland, IDSIA, Switzerland}{monaldo.mastrolilli@supsi.ch}{https://orcid.org/0000-0002-2948-9749}{}

\authorrunning{A. Bortolotti and M. Mastrolilli} 

\Copyright{Alex Bortolotti and Monaldo Mastrolilli} 



\category{} 

\relatedversion{}

\funding{Supported by the Swiss National Science Foundation project n. 200021\_207429 / 1 ``Ideal Membership Problems and the Bit Complexity of Sum of Squares Proofs''}

\nolinenumbers 

\newcommand{\sos}{\text{\sc{SoS}}}
\newcommand{\PS}{\mathcal{P}}
\newcommand{\GIdeal}[1]{\left\langle #1 \right\rangle}

\newcommand{\CSP}{\textsc{CSP}}
\newcommand{\IMP}{\textsc{IMP}}
\newcommand{\Pol}{\mathsf{Pol}}
\newcommand{\Inv}{\mathsf{Inv}}
\newcommand{\PC}{\textsc{PC}}

\newcommand{\Ptime}{\ensuremath{\mathsf{P}}}
\newcommand{\NP}{\ensuremath{\mathsf{NP}}}
\newcommand{\coNP}{\ensuremath{\mathsf{coNP}}}
\newcommand{\ExpSpace}{\ensuremath{\mathsf{ExpSpace}}}

\newcommand{\lex}{\textsf{lex }}
\newcommand{\grlex}{\textsf{grlex }}

\newcommand{\BoolD}{\{0,1\}}
\newcommand{\pin}[1]{\mathsf{Pin}_{#1}}
\newcommand{\Gpin}{\Gamma^{\mathrm{pin}}}

\newcommand{\Variety}[1]{{\textbf{V}}\left( #1 \right)}
\newcommand{\Geq}{\Gamma \cup \{=_{\Lambda}\}}

\usepackage{algpseudocodex}
\newcommand{\runin}[1]{\par\smallskip\noindent\textbf{#1}\quad}

\begin{document}
\maketitle

\begin{abstract}
    The Ideal Membership Problem ($\IMP$) asks whether a polynomial $f$ belongs to an ideal $\mathrm{I}=\langle p_1,\dots,p_m\rangle\subseteq\mathbb{Q}[x_1,\dots,x_n]$. Such membership can be certified in \emph{Polynomial Calculus} ($\PC$), the proof system that derives $f$ from the generators $p_1,\dots,p_m$ by iteratively combining them; the resulting derivation is the certificate. The \emph{degree} of the derivation is its central complexity parameter, and bounding it to $d$ upper bounds the number of derivation steps at $n^{O(d)}$. We write $\PC$-$\IMP_d$ for the resulting problem of producing a degree-bounded $\PC$ certificate of membership, and call it solvable when such a certificate is guaranteed to exist and to be findable in time $n^{O(d)}$. Over a finite field the step bound already yields a polynomial-time algorithm. Over $\mathbb{Q}$, which is the setting we consider, a derivation may instead require exponentially many bits, so it cannot be written down in polynomial time.

    In this paper we study the power of $\PC$ for solving $\IMP$ instances arising from constraint satisfaction problems, and more generally ask for which constraint languages $\Gamma$ such instances can be solved efficiently. Our main contribution is an algebraic reduction framework for $\PC$-$\IMP_d$, based on pp-definitions, pp-interpretations, and pp-encodings, that mirrors the algebraic approach to $\CSP$ complexity. The framework shows that solvability is preserved by these constructions and, in the language of algebras, by passing to subalgebras, finite direct powers, and homomorphic images. We use it to obtain new tractable classes over ternary and larger domains: every language closed under the median operation on a finite chain has solvable $\PC$-$\IMP_d$, by reduction to the Boolean majority algebra, and in particular so does every language over $\{0,1,2\}$ closed under a fixed-value majority. This also places $\IMP_d(\Gamma)$ in $\Ptime$ for such languages, advancing the classification of $\IMP_d$ over ternary domains begun by Mastrolilli and continued by Bulatov and Rafiey. In the process, we settle the last open case of the Boolean dichotomy for $\IMP_d(\Gamma)$ and complete the Boolean classification of $\PC$-$\IMP_d(\Gamma)$ as well: in doing so, we provide an unconditional lower bound for an instance of $\PC$-$\IMP_1$. Finally, a recent $\PC$-to-$\sos$ simulation reduces the \emph{degree-automatability} of Sum-of-Squares---the well-known open problem of finding a degree-$d$ $\sos$ proof in time $n^{O(d)}$ when one exists---to solvability of $\PC$-$\IMP_d$. Each of our new tractable classes therefore yields a family of constraint systems on which $\sos$ proofs are degree-automatable.
\end{abstract}

\newpage

\section{Introduction}

The \emph{Ideal Membership Problem} ($\IMP$) asks whether a polynomial $f \in \mathbb{Q}[x_1,\dots,x_n]$ belongs to the \emph{ideal} generated by $p_1,\dots,p_m \in \mathbb{Q}[x_1,\dots,x_n]$, that is, to the set $\langle p_1,\dots,p_m\rangle = \{\sum_i g_i\,p_i : g_i \in \mathbb{Q}[x_1,\dots,x_n]\}$ of their polynomial combinations \cite{CoxLO15}. A choice of $g_i$ with $f=\sum_i g_i\,p_i$ is a \emph{certificate} of membership, and two questions about such certificates are classical: how large must they be, and can one be found efficiently? The first is the subject of the effective Nullstellensatz, which bounds the certificate degree in terms of the degrees of the generators \cite{Brownawell87,Kollar88}. The second is computational, and is the one we pursue: even deciding membership is \ExpSpace-complete \cite{Mayr89,MayrM82}.

Since the general problem is intractable, we ask for a different certificate: a bounded-degree derivation in \emph{Polynomial Calculus} ($\PC$) rather than the cofactors $g_i$. A $\PC$ derivation builds the elements of an ideal $\mathrm{I}$ from its generators $p_1,\dots,p_m$ by two rules: take $\mathbb{Q}$-linear combinations of polynomials already derived, and multiply a derived polynomial by a variable. A derivation ending at $f$ certifies $f\in \mathrm{I}$; its \emph{degree} is the largest degree occurring in the derivation, and its \emph{bit-size} is the number of bits needed to write it down. Bounding the degree to $d$, one can run Buchberger's algorithm under the \grlex order and discard every polynomial of degree above $d$; this, together with a polynomial division step, finds a $\PC$ derivation of any $f$ of degree at most $d$, whenever one exists, in $n^{O(d)}$ arithmetic operations \cite{CleggEI96}. One might expect a fixed degree to make the certificate easy to find; it does not. Over $\mathbb{Q}$ the coefficients in the derivation may require exponentially many bits \cite{Hakoniemi21}, so bounding the degree alone leaves the certificate intractable to produce.

\medskip
\noindent\textbf{\boldmath$\PC$-$\IMP_d$.}\enspace We study the problem of producing such a certificate at a fixed bounded degree $d$, which we call $\PC$-$\IMP_d$: given $p_1,\dots,p_m,f\in\mathbb{Q}[x_1,\dots,x_n]$ with $\deg f\le d$ and the promise $f\in\langle p_1,\dots,p_m\rangle$, output a $\PC$ derivation of $f$ from $\{p_1,\dots,p_m\}$; the promise makes the task meaningful, since no derivation exists otherwise. Since $d$ is a constant rather than part of the input, $n^{O(d)}$ is polynomial in $n$.
\medskip

\noindent We call $\PC$-$\IMP_d$ \emph{solvable} when every promise instance admits a degree-$O(d)$ derivation \emph{and} this is computable in time $n^{O(d)}$. As mentioned earlier, solvability can fail because no degree-$O(d)$ certificate exists at all, and it can fail because every such certificate has superpolynomial bit-size and so cannot be written down in time $n^{O(d)}$. Note that over a finite field the latter cannot occur.
A $\PC$ derivation of the constant $1$ certifies that $\{ p_i = 0 \}_{i}$ has no common zero; such a derivation is a \emph{refutation}, and it is precisely the degree-$0$ case $\PC$-$\IMP_0$. Algebraic proof complexity typically studies $\PC$ through such refutations. The problem $\PC$-$\IMP_d$ is the membership generalization: the target is an arbitrary $f$.

We investigate $\PC$-$\IMP_d$ when the generators encode an instance of a \emph{constraint satisfaction problem} ($\CSP$). A \emph{constraint language} $\Gamma$ over a finite domain $\Lambda$ is a finite set of relations, a \emph{constraint} applies one of them to a tuple of variables, and $\CSP(\Gamma)$ asks whether a system of constraints admits a satisfying assignment; both $\Gamma$ and $\Lambda$ are fixed throughout. To an instance $\mathcal{T}$ one associates a generating set $\mathcal{F}$ containing, for each constraint $C$, an interpolant $f_C$ vanishing exactly on the tuples allowed by its relation, together with domain polynomials $f_\Lambda(x_i)=\prod_{\lambda\in\Lambda}(x_i-\lambda)$ pinning each variable to $\Lambda$ \cite{Mastrolilli21}. The domain polynomials restrict the zeros of $\mathcal{F}$ to $\Lambda^n$, where they are exactly the satisfying assignments of $\mathcal{T}$; by the Weak Nullstellensatz, the instance is unsatisfiable precisely when $1\in\langle \mathcal{F}\rangle$. The degree-$0$ case $\PC$-$\IMP_0(\Gamma)$ thus coincides with finding $\PC$ refutations of unsatisfiable $\CSP(\Gamma)$ instances. The algebraic approach to $\CSP$s classifies the hardness of $\CSP(\Gamma)$ by the \emph{polymorphisms} of $\Gamma$---the operations under which all its relations are closed---and this viewpoint famously led to a complete dichotomy: every $\CSP(\Gamma)$ is in $\Ptime$ or $\NP$-complete \cite{Bulatov17,Zhuk20}. The same approach has proved fruitful for ideal membership: the polymorphism-based study of $\IMP$ was initiated by Mastrolilli \cite{Mastrolilli19,Mastrolilli21}, who delineated the Boolean tractability borderline for the restriction of $\IMP$ to a fixed constraint language $\Gamma$---which we write $\IMP(\Gamma)$---with later extensions to larger domains \cite{BharathiM22,BharathiM25,BortolottiMV25,BulatovRSTOC22}.

\runin{Membership versus its certificate.}
Demanding an explicit $\PC$ certificate is strictly harder than deciding membership. The sharpest witnesses are the unsatisfiable affine systems over $\mathbb{F}_2$: for the languages $\Gamma$ that produce them, $\IMP_d(\Gamma)\in\Ptime$ \cite{BulatovRSTOC22,Mastrolilli21}, and yet $\PC$ refutations over $\mathbb{Q}$ may require degree $\Omega(n)$ \cite{BussGIP01} (recast in the algebraic $\CSP$ framework by \cite{AtseriasO19}), so that $\PC$-$\IMP_0(\Gamma)$ admits no constant-degree solution. The existing $\IMP$ classification therefore does not settle $\PC$-$\IMP_d(\Gamma)$: its hardness transfers only under $\Ptime \neq \NP$, whereas $\PC$ admits unconditional bounds. $\PC$-$\IMP_d$ calls for a classification of its own.

\runin{Sum-of-Squares automatability.}
Whether fixed-degree Sum-of-Squares ($\sos$) proofs are \emph{degree-automatable} (findable in time $n^{O(d)}$ whenever a degree-$d$ proof exists) is a well-known open problem: a system may admit a low-degree $\sos$ proof, yet every such proof can require coefficients of doubly-exponential magnitude in $n$, putting them out of reach of the standard semidefinite programming methods such as the ellipsoid method \cite{Hakoniemi21, Odonnell2017, RaghavendraW17}. The question remains under active investigation from several directions \cite{BortolottiMPV26,GriblingPS23,GartnerMV26,PalombaSVM25}, and, more broadly, in convex optimization \cite{SlotSW26}.
The $\PC$ criterion of \cite{BortolottiMV25} ties $\PC$-$\IMP_d$ directly to this question. Building on a simulation of $\PC$ by $\sos$ \cite{Berkholz18,BortolottiMV25}, it shows that whenever $\PC$-$\IMP_d$ is solvable for a system $\PS$, any degree-$2d$ $\sos$ proof of ``$r\ge0$'' from $\PS$ can be replaced, up to an arbitrary additive $\varepsilon$, by one of degree $O(d)$ whose coefficients are bounded by $2^{\mathrm{poly}(n^d)}$. The criterion thus reduces $\sos$ degree-automatability to $\PC$-$\IMP_d$: solvable instances yield automatable $\sos$ proofs, including for the new tractable classes we obtain below.

\runin{Our contributions.}
We develop an algebraic reduction framework for $\PC$-$\IMP_d$ mirroring the polymorphism approach to $\CSP$ complexity, and apply it to obtain new tractable classes. The framework rests on a single observation: solvability of $\PC$-$\IMP_d$ is preserved by the standard constructions relating constraint languages. Over a fixed domain, if $\Gamma$ pp-defines $\Delta$ and $\PC$-$\IMP_d(\Gamma)$ is solvable, then so is $\PC$-$\IMP_d(\Delta)$ (\cref{th:pp-definability and PC reductions}); across domains, the same holds for pp-interpretations and pp-encodings (\cref{th:pp-interpretations and PC reductions}, \cref{th:pp-encodings and PC reductions}). What these statements add to their $\IMP$ counterparts is that each substitution step must itself be carried out by a bounded-degree derivation whose coefficients stay of polynomial bit-size, which is what our proofs supply; \cref{def:PC reduction} makes this notion of reduction precise. Recast in the language of algebras, solvability then propagates along subalgebras, finite direct powers, and homomorphic images (\cref{sect:algebras}).

The framework yields new tractable classes over ternary and larger domains. On a finite chain $\{0,1,\dots,\ell\}$, the median operation $m_\ell$ reduces to the Boolean majority algebra $\mathcal{B}=(\{0,1\},\textsc{Maj})$: the algebra $(\{0,\dots,\ell\},m_\ell)$ is an isomorphic image of a subalgebra of $\mathcal{B}^\ell$, so, since $\PC$-$\IMP_d(\Gamma)$ is solvable for every $\Gamma$ closed under $\textsc{Maj}$ \cite{BortolottiMV_ARXIV25}, the transfer along subalgebras, powers, and homomorphic images applies (\cref{prop:median-chain}). As the special case $\ell=2$, every constraint language over $\{0,1,2\}$ closed under a \emph{fixed-value majority} $\mu_i$---the majority that returns a constant on all-distinct triples---has solvable $\PC$-$\IMP_d$ for every fixed $d$ (\cref{cor:012-majorities}). This in turn places $\IMP_d(\Gamma)$ in $\Ptime$ for such $\Gamma$, advancing the classification of $\IMP_d$ over ternary domains begun by Mastrolilli~\cite{Mastrolilli21} and continued by Bulatov and Rafiey~\cite{BulatovRSTOC22} (\cref{sect:majorities}). Through the $\PC$ criterion of \cite{BortolottiMV25}, these tractable classes also yield new families of degree-automatable $\sos$ proofs.

Along the way we complete the Boolean picture for both problems (\cref{sect:csp-pc-imp,th:complete Boolean classification}). For the decision problem $\IMP_d(\Gamma)$, Mastrolilli \cite{Mastrolilli21} left one case open: whether $\IMP_1(\Gamma)$ is $\coNP$-complete when $\Gamma$ is closed under the two constants $\mathbf{0}$ and $\mathbf{1}$, but under none of $\neg$, $\textsc{Majority}$, $\textsc{Minority}$, $\textsc{Max}$, $\textsc{Min}$. We settle it in the affirmative by reducing $\CSP(\Gpin)$, with $\Gpin := \Gamma\cup\{\pin{0},\pin{1}\}$, to the complement of $\IMP_1(\Gamma)$; for such $\Gamma$, $\CSP(\Gpin)$ is $\NP$-complete by Schaefer's dichotomy \cite{Jeavons98,Schaefer78} (\cref{prop:degree-one-hardness}).

The $\PC$-$\IMP$ dichotomy does not follow from this one. Its tractable side is strictly smaller, and rests on $\PC$ constructions of its own \cite{AtseriasO19,BortolottiMV25}, as the $\IMP_d$ algorithms do not always work within bounded-degree $\PC$. What transfers from the hard side is conditional: $\coNP$-hardness of $\IMP_d(\Gamma)$ rules out solving $\PC$-$\IMP_d(\Gamma)$ unless $\Ptime = \NP$, and it says nothing about the $\textsc{Minority}$ languages closed under a constant operation, where every instance is satisfiable and $\IMP_d(\Gamma) \in \Ptime$ for every $d$. We settle that case unconditionally, since a substitution turns any $\PC$ derivation of a suitable degree-$1$ polynomial into a refutation of a Tseitin system over an expander (\cref{prop:affine degree one}).

\runin{Related work.}
Closest to our reduction framework is the work of Bulatov and Rafiey \cite{BulatovR26}, who develop a reduction framework based on interpolation, similar to ours, for the closely related problem $\chi\IMP_d$: given an ideal $\mathrm{I}$ and polynomials $p_1, \ldots, p_m$, decide whether some linear combination of them belongs to $\mathrm{I}$. They aptly name their technique \emph{reduction by substitution}. Their framework also covers reductions by pp-definitions and pp-interpretations and, together with \cite{BulatovRSTOC22}, the corresponding algebraic reductions. Notably, solving $\chi\IMP_d$ efficiently lets them construct $d$-truncated Gr\"obner bases semantically: candidates are selected one at a time and tested for membership in the ideal by checking whether the candidate vanishes over the relevant variety. It remains however unclear whether such an algorithm can be simulated by $\PC$. Our question of interest is whether a reduction is viable \emph{syntactically}, that is, at the level of $\PC$ derivations, which, we emphasize, is what our results need in order to propagate to $\sos$ automatability.

A concept closely related to $\PC$-$\IMP_d$, and extensively studied in proof complexity, is \emph{automatability} \cite{BonetPR00}: a proof system is automatable if a proof can be found in time polynomial in the size of the \emph{shortest} proof. In this strong sense the question is settled negatively for $\PC$, since automating it is $\NP$-hard \cite{deRezendeGNPRS21}; the weaker degree-automatability is, as discussed above, more delicate. These results treat $\PC$ as a \emph{refutation} system, as do the classical degree lower bounds \cite{BussGIP01,Razborov98}.
In the context of refutations, Atserias and Ochremiak \cite{AtseriasO19} show that the standard $\CSP$ reductions preserve proof complexity across the main proof systems, giving an algebraic characterization of which languages admit small certificates, with $\CSP(\Gamma)$ presented through homomorphisms of relational structures.
In the context of $\PC$-$\IMP_d$ derivations of arbitrary targets $f$, the bounded-degree viewpoint already appears in Jefferson et al.\ \cite{JeffersonJGD13}, whose polynomial-time truncated Gr\"obner computations simulate the local-consistency (propagation) algorithms of constraint programming.

\section{Preliminaries}\label{sect:preliminaries}

\runin{Polynomial Calculus.} 
Given a set of polynomial equality constraints $\mathcal{P} = \{ p_1 = 0, \ldots, p_m = 0 \} \subseteq \mathbb{Q}[x_1, \ldots, x_n]$, $\PC$ is characterized by the inference rules: let $p_i \in \mathcal{P}$, $f,g \in \mathbb{Q}[x_1, \ldots, x_n]$ and $a,b \in \mathbb{Q}$
\begin{align}\label{eq:PC rules}
    \frac{}{p_i=0} \qquad \frac{f=0 \qquad g=0}{a f + b g =0} \qquad \frac{f=0}{x_j f=0} 
\end{align}
called \emph{axiom}, \emph{addition} and \emph{multiplication} rules, respectively. A $\PC$ derivation of $``q = 0"$ from $\mathcal{P}$ is a sequence $\mathfrak{Q} = [q_1, \ldots, q_{s-1}, q]$ in which each $q_i$ follows from \cref{eq:PC rules}; it certifies that $q$ lies in the ideal generated by $\mathcal{P}$ and, by soundness, that $q$ vanishes on the \emph{feasibility set} $V := \{x \mid p_i(x) = 0 \ \forall i \in [m]\}$. A derivation of $``1 = 0"$ is a $\PC$ \emph{refutation}, certifying $V = \emptyset$. The \emph{degree} of $\mathfrak{Q}$ is the greatest degree appearing in it, and its \emph{bit-size} the number of bits needed to write it down.

Throughout, polynomials are in $\mathbb{Q}[x_1, \ldots, x_n]$, the ring is ordered by \emph{graded lexicographic} (\textsf{grlex}) order, and the proof system is Polynomial Calculus over $\mathbb{Q}$ (usually denoted $\PC/\mathbb{Q}$; we consider no other field and write $\PC$).

\runin{$\CSP$s and polymorphisms.}
Let $\Lambda \subseteq \mathbb{Q}$ be a finite domain and $X = \{x_1, \ldots, x_n\}$ a set of variables. A \emph{$k$-ary relation} is a set $R \subseteq \Lambda^k$, and a \emph{constraint language} $\Gamma$ is a finite set of relations. A \emph{constraint} over $\Gamma$ is a pair $((x_{i_1}, \ldots, x_{i_k}), R)$ with $R \in \Gamma$ a $k$-ary relation, also written $R(x_{i_1}, \ldots, x_{i_k})$. An instance $\mathcal{T}$ of $\CSP(\Gamma)$ is a triple $(X, \Lambda, \mathcal{C})$ with $\mathcal{C}$ a set of constraints; it asks whether some assignment $\phi \colon X \to \Lambda$ satisfies $(\phi(x_{i_1}), \ldots, \phi(x_{i_k})) \in R$ for every constraint $R(x_{i_1}, \ldots, x_{i_k}) \in \mathcal{C}$. 
An \emph{$m$-ary polymorphism} of $\Gamma$ is an operation $\varphi \colon \Lambda^{m} \to \Lambda$ such that, for every $R \in \Gamma$, applying $\varphi$ componentwise to any $m$ tuples of $R$ (repetitions allowed) yields a tuple of $R$. The set of all polymorphisms of $\Gamma$ is denoted $\Pol(\Gamma)$.

To a constraint $C = R(x_{i_1},\ldots, x_{i_k})$ we associate the (Lagrange) interpolating polynomial $f_C \in \mathbb{Q}[x_1, \ldots, x_n]$ with $f_C(a) = 0 \iff a \in R$ for every $a \in \Lambda^k$, together with the \emph{domain polynomial} $f_{\Lambda}(x_i) = \prod_{\lambda \in \Lambda} (x_i - \lambda)$. See e.g.\ \cite{BulatovRSTOC22,Mastrolilli21} for details.

\runin{On radical ideals.}
We collect some facts used below (see e.g.\ \cite{BeckerW93,CoxLO15}). An ideal $\mathrm{I}$ is \emph{radical} if $p^d \in \mathrm{I}$ implies $p \in \mathrm{I}$. The \emph{variety} $\mathbf{V}(\mathcal{P})$ is the set of points at which every polynomial in $\mathcal{P}$ vanishes; over an algebraically closed field the Weak Nullstellensatz gives $\mathbf{V}(\mathrm{I}) = \emptyset \iff 1 \in \mathrm{I}$. The ideals of interest are generated by axioms together with the domain polynomials $f_\Lambda(x_i)$, and for these the standard Nullstellensatz holds already over $\mathbb{Q}$:
\begin{lemma}[Finite-domain Nullstellensatz]\label{th:domain ideal nullstellensatz}
    Let $\Lambda \subseteq \mathbb{Q}$ be a finite domain, $\mathcal{F} \subseteq
    \mathbb{Q}[x_1,\ldots,x_n]$, and $\mathrm{I} = \langle \mathcal{F} \cup
    \{f_{\Lambda}(x_i)\}_{i \in [n]} \rangle$. Then $\mathrm{I}$ is radical
    \textup{\cite[Lemma 8.13]{BeckerW93}}, $\mathbf{V}(\mathrm{I}) \subseteq \Lambda^n$, and
    for every $f \in \mathbb{Q}[x_1,\ldots,x_n]$,
    \[
        f \in \mathrm{I} \iff f(a) = 0 \ \text{ for all } a \in \mathbf{V}(\mathrm{I}); \qquad \textit{(Strong Nullstellensatz)}
    \]
    in particular $1 \in \mathrm{I} \iff \mathbf{V}(\mathrm{I}) = \emptyset$ \textit{(Weak Nullstellensatz)}.
\end{lemma}

\runin{The problem $\PC$-$\IMP_d(\Gamma)$.}
Given a constraint language $\Gamma$ over $\Lambda$, an instance of $\PC$-$\IMP_d(\Gamma)$ is a pair $(f, \mathcal{T})$ with $f \in \mathbb{Q}_d[x_1,\ldots,x_n]$ of bit-size $n^{O(d)}$ and $\mathcal{T}$ an instance of $\CSP(\Gamma)$ with constraint set $\mathcal{C}$. We write $\mathcal{F} := \{f_C\}_{C \in \mathcal{C}} \cup \{f_{\Lambda}(x_i)\}_{i \in [n]}$ for the associated set of generators (\emph{axioms}); the domain polynomials always belong to $\mathcal{F}$, so $\langle \mathcal{F} \rangle$ is radical by \cref{th:domain ideal nullstellensatz}.
We emphasize that the generators of a $\PC$-$\IMP_d(\Gamma)$ instance have degree $O(d)$ and bit-size $n^{O(d)}$ (polynomial in $n$ for fixed $d$); this holds since $\Gamma$ is finite and the $f_C$ are interpolating polynomials.
The input comes with the promise $f \in \langle \mathcal{F} \rangle$, equivalently---since $\PC$ derives from $\mathcal{F}$ exactly the members of $\langle \mathcal{F} \rangle$ (soundness and completeness)---that a $\PC$ derivation of $f$ from $\mathcal{F}$ exists at some degree depending on $n$; the task is to output one. Closely related is the \emph{decision} problem $\IMP_d(\Gamma)$: given such a pair $(f, \mathcal{T})$, decide whether $f \in \langle \mathcal{F} \rangle$. We work throughout over arbitrary finite domains $\Lambda$, not only the Boolean case.

Since $\langle \mathcal{F} \rangle$ is radical and $f_{\Lambda}$ restricts each $x_i$ to $\Lambda$, \cref{th:domain ideal nullstellensatz} identifies $\mathbf{V}(\mathcal{F})$ with the solution set of $\mathcal{T}$, so the promise $f \in \langle \mathcal{F} \rangle$ amounts to $f$ vanishing on all solutions of the $\CSP(\Gamma)$ instance.

\subsection{Polynomial encodings}\label{sect:encodings}
A consequence of radicality is that the particular choice of polynomial encoding does not matter: the degree and size of $\PC$ derivations are then independent of the encoding. Fix a $\CSP(\Gamma)$ instance $\mathcal{T}$ with constraint set $\mathcal{C}$. A \emph{polynomial encoding} of $\mathcal{T}$ is any generator set $\{g_C\}_{C \in \mathcal{C}} \cup \{f_{\Lambda}(x_i)\}_{i \in [n]}$ whose $g_C$ have degree $O(1)$ and vanish exactly on the tuples of their relation; in particular, $\mathcal{F}$ is a polynomial encoding.

\begin{proposition}\label{prop:encoding-independence}
    Let $\mathcal{F} = \{f_C\}_{C} \cup \{f_{\Lambda}(x_i)\}_i$ and $\mathcal{G} = \{g_C\}_{C} \cup \{f_{\Lambda}(x_i)\}_i$ be two polynomial encodings of the same instance $\mathcal{T}$, of bit-size $n^{O(d)}$. Then every $\PC$ derivation of $f$ from $\mathcal{F}$ of degree $\delta$ and bit-size $b$ converts into a $\PC$ derivation of $f$ from $\mathcal{G}$ of degree $\max\{\delta, O(1)\}$ and bit-size $b + |\mathcal{C}|\,n^{O(d)}$, and symmetrically.
\end{proposition}

    \begin{proof}
    For each constraint $C = R(x_{i_1}, \ldots, x_{i_k})$, both $\{f_C\} \cup \{f_{\Lambda}(x_{i_j})\}_{j \in [k]}$ and $\{g_C\} \cup \{f_{\Lambda}(x_{i_j})\}_{j \in [k]}$ have common zero set $R$ over $\Lambda^k$, so by \cref{th:domain ideal nullstellensatz} they generate the same ideal, namely the set of polynomials vanishing on $R$. In particular $f_C \in \langle g_C, \{f_{\Lambda}(x_{i_j})\}_{j \in [k]} \rangle$, an ideal presenting polynomials with $O(1)$ variables and degree, and bit-size $n^{O(d)}$, so by the following procedure $f_C$ admits a $\PC$ derivation from $\mathcal{G}$ of degree $O(1)$ and bit-size $n^{O(d)}$.
    \begin{algorithmic}[1]
        \Statex Input : $f_C, \mathcal{G}_C = \{g_C, f_{\Lambda}(x_{i_1}), \ldots, f_{\Lambda}(x_{i_k})\}$.
        \Statex Output : a $\PC$-derivation $\mathcal{Q} = [q_1, \ldots, q_s = f_C]$ from $\mathcal{G}_C$.
        \Statex
        \State $\mathcal{H} := \textbf{Buchberger}(\mathcal{G}_C)$
        \State $\mathcal{Q} := \textbf{Step-by-Step-Buchberger}(\mathcal{G}_C)$
        \State $\mathcal{Q} := \mathcal{Q} + \textbf{Reverse}(\textbf{Step-by-Step-PolynomialDivision}(f_C, \mathcal{H}))$
        \State \Return $\mathcal{Q}$
    \end{algorithmic}
    Here \textbf{Buchberger} and \textbf{Step-by-Step-Buchberger} both run the standard Buchberger algorithm (see e.g. \cite{CoxLO15}); the former returns only the Gröbner basis of $\langle \mathcal{G}_C \rangle$, the latter the sequence of polynomials produced during the run, as a $\PC$ derivation. The subprocedure \textbf{Step-by-Step-PolynomialDivision} is analogous, and \textbf{Reverse} flips the list so that $\mathcal{Q}$ terminates with $f_C$.

    The algorithm relies essentially on two procedures: forming an S-polynomial and polynomial division. Both only add polynomials and multiply them by monomials, which are $\PC$ operations. Crucially, since $\Gamma$ is finite, the arity $k$ is bounded by a constant, so this subsystem of $\mathcal{G}$ has $O(1)$ variables and the Gröbner computation terminates in $O(1)$ \emph{steps}. Each step is polynomial in the bit-size of the encoding, which is $n^{O(d)}$, so the computation runs in \emph{time} $n^{O(d)}$ and yields a derivation of constant degree with coefficients of bit-size $n^{O(d)}$ (see also \cite{BortolottiMV_ARXIV25}).

    Replacing each axiom line $f_C$ by its derivation from $\mathcal{G}$ turns the original derivation into one from $\mathcal{G}$; the domain polynomials are shared by $\mathcal{F}$ and $\mathcal{G}$ and need no conversion. The result is a $\PC$ derivation of $f$ from $\mathcal{G}$: its degree is the maximum of $\delta$ and the $O(1)$ conversion degrees, and its bit-size exceeds $b$ by $n^{O(d)}$ per constraint, hence by $|\mathcal{C}|\,n^{O(d)}$. The symmetric statement is identical with $\mathcal{F}, \mathcal{G}$ exchanged.
\end{proof}


\section{$\PC$-$\IMP$ Reductions}\label{sect:PC-IMP reductions}

This section develops the reduction machinery that underlies the rest of the paper. The aim is to transfer $\PC$-$\IMP_d$ tractability from one constraint language to another: if we can solve bounded-degree $\PC$-$\IMP$ over $\Gamma$, and $\Gamma$ is related to $\Delta$ in the right algebraic sense, then we can solve it over $\Delta$ as well. The pp-definitions and pp-interpretations we use are standard; see e.g.\ \cite{BartoKW17}. We first make explicit the notion of reduction that both subsections instantiate.

\begin{definition}\label{def:PC reduction}
    Let $(f_{\mathcal{T}}, \mathcal{T})$ and $(f_{\mathcal{S}}, \mathcal{S})$ be instances of $\PC$-$\IMP$, over variables $X$ and $Y$ and with generators $\mathcal{F}_{\mathcal{T}}$ and $\mathcal{F}_{\mathcal{S}}$, respectively. A \emph{substitution of degree $d_1$} is a family $r = \{r_y\}_{y \in Y}$ of polynomials of $\mathbb{Q}[X]$ with $\deg r_y \leq d_1$ and coefficients of bit-size polynomial in $n$; we write $r(h)$ for the polynomial obtained from $h$ by replacing every $y$ with $r_y$.

    We say that $(f_{\mathcal{T}}, \mathcal{T})$ is \emph{$(d_1,d_2)$-reducible} to $(f_{\mathcal{S}}, \mathcal{S})$ if there is a substitution $r$ of degree $d_1$ such that the following admit $\PC$ derivations from $\mathcal{F}_{\mathcal{T}}$, of degree $d_2$ and of bit-size polynomial in $n$:
    \begin{enumerate}
        \item $r(h)$, for every $h \in \mathcal{F}_{\mathcal{S}}$;
        \item $f_{\mathcal{T}} - r(f_{\mathcal{S}})$.
    \end{enumerate}

    Let $d$ and $e$ be fixed degrees. We say that $\PC$-$\IMP_d(\Delta)$ \emph{$\PC$-reduces} to $\PC$-$\IMP_e(\Gamma)$ if every instance of $\PC$-$\IMP_d(\Delta)$ is $(O(1),O(1))$-reducible to some instance of $\PC$-$\IMP_e(\Gamma)$, and a polynomial-time algorithm returns that instance together with the substitution $r$ and derivations of the two polynomials in 1. and 2. above.
\end{definition}

$(d_1,d_2)$-reducibility was already introduced in \cite{BussGIP01}, and the additional polynomial-time map on instances in \cite{GalesiGPS23}. Both are stated for \emph{refutations}. We additionally require a bound on the bit-size. Given a substitution $r$, a degree-$\delta$ derivation of $f_{\mathcal{S}}$ from $\mathcal{F}_{\mathcal{S}}$ becomes a degree-$d_1\delta$ derivation of $r(f_{\mathcal{S}})$ from $r(\mathcal{F}_{\mathcal{S}})$, and $(d_1, d_2)$-reducibility gives a well-behaved $\PC$ derivation of $f_{\mathcal{T}}$ from $\mathcal{F}_{\mathcal{T}}$. A $\PC$ reduction therefore transfers solvability, which is how \cref{th:pp-definability and PC reductions,th:pp-interpretations and PC reductions} are applied below.

\subsection{Pp-definitions}\label{sect:pp-definitions}

In this section we show how the efficiency of solving $\PC$-$\IMP_d$ transfers by means of pp-definitions, which we first recall.

\begin{definition}
    Let $\Gamma, \Delta$ be constraint languages over the same finite domain $\Lambda$. We say that $\Gamma$ \emph{pp-defines} $\Delta$ (or $\Delta$ is \emph{pp-definable} from $\Gamma$) if for each relation (predicate) $R \subseteq \Lambda^{k}$ in $\Delta$ there exists a first order formula $L$ over variables $\{x_1, \ldots , x_m, x_{m+1}, \ldots, x_{m+k}\}$ that uses only constraints (predicates) from $\Gamma$, equality relations, and conjunctions, such that
    \begin{equation}\label{eq:pp-formula}
        R(x_{m+1}, \ldots, x_{m+k}) = \exists x_1 \ldots \exists x_m L.
    \end{equation}
    Such an expression is called a \emph{primitive positive (pp-)formula}.
\end{definition}

We write $=_{\Lambda}$ for the equality relation $\{(a,a) : a \in \Lambda\}$ on $\Lambda$ and encode it by the polynomial $f_{=_{\Lambda}}(x,x') = x - x'$.

Fix $\Gamma, \Delta$ over $\Lambda$ with $\Gamma$ pp-defining $\Delta$, and let $(f_{\Delta}, \mathcal{T}_{\Delta})$ be an instance of $\PC$-$\IMP_d(\Delta)$ with $\mathcal{T}_\Delta = (X, \Lambda, \mathcal{C}_\Delta)$. Write
\begin{align}\label{eq:generators Delta}
    \mathcal{F}_\Delta := \{ \bigcup_{C \in \mathcal{C}_{\Delta}} f_C \} \cup \{f_{\Lambda}(x)\}_{x \in X}
\end{align}
for the generators of $\mathcal{T}_\Delta$. Each constraint $C \in \mathcal{C}_\Delta$ has a pp-formula $C(X_C) = \exists Y_C\, L_C(Y_C, X_C)$, with $L_C$ a conjunction of constraints from $\Gamma$ and equalities on $Y_C \cup X_C$, where the variables of $Y_C$ are taken fresh, so that the $Y_C$ are pairwise disjoint. Collecting the existentially quantified variables $Y = \bigcup_{C} Y_C$ and all the constraints occurring in $L_C$ into $\mathcal{C}_{\Geq}$ yields the instance $\mathcal{T}_{\Geq} := (Y \cup X, \Lambda, \mathcal{C}_{\Geq})$, with generators
\begin{align}\label{eq:generators Gamma}
    \mathcal{F}_{\Geq} := \{\bigcup_{C \in \mathcal{C}_{\Geq}} f_C \} \cup \{f_{\Lambda} (x) \}_{x \in Y \cup X}.
\end{align}
Setting $f_{\Geq} := f_\Delta$, which we denote simply by $f$, we obtain the instance $(f, \mathcal{T}_{\Geq})$, which however is not an instance of $\PC$-$\IMP_d(\Gamma)$ as the equality constraints of the pp-formulas are not constraints of $\Gamma$. 

We therefore collapse them. Specifically, consider the partition of $Y \cup X$ by means of equalities of $\mathcal{T}_{\Geq}$, let $\sigma$ be the map that sends every variable of $Y \cup X$ to a fixed representative of its class, chosen in $X$ whenever the class meets $X$, and let $\mathcal{T}_{\Gamma} := (\sigma(Y \cup X), \Lambda, \mathcal{C}_{\Gamma})$ be the instance obtained by substituting $\sigma(x)$ for $x$ and removing the equality constraints. We obtain the generators
\begin{align}\label{eq:generators Gamma collapsed}
    \mathcal{F}_{\Gamma} := \{\bigcup_{C \in \mathcal{C}_{\Gamma}} f_C \} \cup \{f_{\Lambda}(x)\}_{x \in \sigma(Y \cup X)}.
\end{align}
The underlying instance $\mathcal{T}_{\Gamma}$ is an instance of $\CSP(\Gamma)$. Lastly, we consider the polynomial $\sigma(f)$, of degree at most $d$ and, by the choice of the representatives, in the variables of $X$.

\begin{lemma}[\cite{BulatovRSTOC22}]\label{th:pp-definability IMP reduction}
    $f \in \langle \mathcal{F}_{\Delta} \rangle$ if and only if $\sigma(f) \in \langle \mathcal{F}_{\Gamma} \rangle$.
\end{lemma}

\begin{proof}
    By \cref{th:domain ideal nullstellensatz} the ideals $\langle \mathcal{F}_{\Delta} \rangle$ and $\langle \mathcal{F}_{\Gamma} \rangle$ are radical, with $\Variety{\mathcal{F}_{\Delta}}$ and $\Variety{\mathcal{F}_{\Gamma}}$ the solution sets of $\mathcal{T}_{\Delta}$ and of $\mathcal{T}_{\Gamma}$. It thus suffices to show that $f$ vanishes on $\Variety{\mathcal{F}_{\Delta}}$ if and only if $\sigma(f)$ vanishes on $\Variety{\mathcal{F}_{\Gamma}}$.

    Given $b \in \Lambda^{|Y \cup X|}$ we write $b|_{X}$ for its restriction to the coordinates of the variables in $X$, and given $c \in \Lambda^{|\sigma(Y \cup X)|}$ we write $\tilde{c} \in \Lambda^{|Y \cup X|}$ for the tuple obtained from $c$ by repeating each coordinate along the class of $\sigma$ it represents. Then $\sigma(f)(c) = f(\tilde{c}|_{X})$, as $\sigma(f)$ is obtained from $f$ by substituting $\sigma(x)$ for $x$, where $f$ involves only variables of $X$.

    Assume $f \in \langle \mathcal{F}_{\Delta} \rangle$ and consider $c \in \Variety{\mathcal{F}_{\Gamma}}$. By construction $\tilde{c}$ satisfies every equality constraint of $\mathcal{T}_{\Geq}$, whereas every other constraint of $\mathcal{T}_{\Geq}$ is mapped by $\sigma$ to a constraint of $\mathcal{T}_{\Gamma}$, and so satisfied by $c$. Hence $\tilde{c} \in \Variety{\mathcal{F}_{\Geq}}$ and, since each $C \in \mathcal{C}_{\Delta}$ is satisfied as soon as $L_C$ is, $\tilde{c}|_{X} \in \Variety{\mathcal{F}_{\Delta}}$. We conclude that $\sigma(f)(c) = f(\tilde{c}|_{X}) = 0$.

    Conversely, assume $\sigma(f) \in \langle \mathcal{F}_{\Gamma} \rangle$ and consider $a \in \Variety{\mathcal{F}_{\Delta}}$. As $a$ satisfies every $C \in \mathcal{C}_{\Delta}$, we may extend it to a tuple $b$ with $b|_{X} = a$ by choosing for each such $C$ any tuple of witnesses for the variables $Y_C$ of its pp-formula. Every constraint of $L_C$ is then satisfied, equalities included, so that $b \in \Variety{\mathcal{F}_{\Geq}}$ and, since the classes of $\sigma$ are generated by the equality constraints, $b$ repeats each coordinate along the class of $\sigma$ it belongs to. Letting $c$ be the restriction of $b$ to the coordinates of the representatives we thus have $b = \tilde{c}$, where $c \in \Variety{\mathcal{F}_{\Gamma}}$, since every constraint of $\mathcal{C}_{\Gamma}$ is the image under $\sigma$ of a constraint of $\mathcal{C}_{\Geq}$ satisfied by $b$. Therefore $f(a) = f(\tilde{c}|_{X}) = \sigma(f)(c) = 0$.
\end{proof}

The above lemma supplies the polynomial-time map on instances. In the rest of this subsection we show that this is indeed a $\PC$ reduction.

Since $\sigma(f) \in \langle \mathcal{F}_{\Gamma} \rangle$, the pair $(\sigma(f), \mathcal{T}_{\Gamma})$ is an instance of $\PC$-$\IMP_d(\Gamma)$. Two steps then take any degree-$O(d)$ $\PC$ derivation of $\sigma(f)$ from $\mathcal{F}_{\Gamma}$ to one of $f$ from $\mathcal{F}_{\Delta}$: (a) we have to eliminate the variables of $Y$, and (b) to retrieve a derivation of $f$ from the derivation of $\sigma(f)$ from generators of $\Gamma$ (as opposed to images of generators by $\sigma$). We take care of the latter first.

\begin{lemma}\label{th:renaming}
    Let $\mathcal{F} \subseteq \mathbb{Q}[x_1, \ldots, x_n]$, let $x, x'$ be variables such that $x - x'$ has a $\PC$ derivation from $\mathcal{F}$ of degree $\delta_0$ and bit-size $s_0$, and let $\tau$ denote the substitution of $x$ by $x'$. Then for every $g \in \mathbb{Q}[x_1, \ldots, x_n]$ of degree at most $D$ the polynomial $g - \tau(g)$ has a $\PC$ derivation from $\mathcal{F}$ of degree $\max\{\delta_0, D\}$ and bit-size polynomial in $s_0$, in the bit-size of $g$, and in $D$.
\end{lemma}

\begin{proof}
    Write $g = \sum_{t \geq 0} c_t\, x^t$ with $c_t \in \mathbb{Q}[x_1, \ldots, x_n]$ not involving $x$, so that $\tau(g) = \sum_{t \geq 0} c_t\, (x')^t$ and
    \begin{align}\label{eq:substitution divisibility}
        g - \tau(g) \; = \; \sum_{t \geq 1} c_t \, \big(x^t - (x')^t\big) \; = \; (x - x') \cdot h, \qquad h \; := \; \sum_{t \geq 1} c_t \sum_{u = 0}^{t-1} x^{u} (x')^{t-1-u}.
    \end{align}
    This exhibits $g - \tau(g)$ as a multiple of $x - x'$. The claimed $\PC$ derivation follows accordingly.
\end{proof}

Every equality constraint of $\mathcal{T}_{\Geq}$ contributes to $\mathcal{F}_{\Geq}$ an axiom of degree $1$. Hence, for any two variables $x, x'$ such that $\sigma(x) = \sigma(x')$, summing these axioms along a path joining $x$ to $x'$ yields a degree-$1$ derivation of $x - x'$ from $\mathcal{F}_{\Geq}$. Therefore \cref{th:renaming} can be applied both to the generators of $\mathcal{F}_{\Gamma}$ and to $\sigma(f)$, proving the following.

\begin{lemma}\label{th:equality collapse}
    Given a degree-$O(d)$ $\PC$ derivation of $\sigma(f)$ from $\mathcal{F}_{\Gamma}$ of bit-size $b$, a degree-$O(d)$ $\PC$ derivation of $f$ from $\mathcal{F}_{\Geq}$ can be computed in time polynomial in $b$ and $n$.
\end{lemma}

\begin{proof}
    Write $Z := Y \cup X$ for the variables of $\mathcal{T}_{\Geq}$ and consider the images under $\sigma$ of the generators of $\mathcal{F}_{\Geq}$; every equality axiom is mapped to $0$, leaving $\{\sigma(f_C)\}_{C} \cup \{f_{\Lambda}(\sigma(z))\}_{z \in Z}$. Each $\sigma(f_C)$ vanishes on $\Lambda^{|\sigma(Z)|}$ exactly at the tuples satisfying the collapsed constraint, and has degree $O(1)$, so this is a polynomial encoding of $\mathcal{T}_{\Gamma}$ in the sense of \cref{sect:encodings}. By \cref{prop:encoding-independence} we may therefore take it as $\mathcal{F}_{\Gamma}$.

    As observed before \cref{th:equality collapse}, \cref{th:renaming} applies with $\delta_0 = 1$ to any pair of variables having the same image under $\sigma$. Renaming the non-representative variables of a polynomial $g$ one at a time and summing the resulting chain of differences, we therefore obtain, for every $g$, a derivation of $g - \sigma(g)$ from $\mathcal{F}_{\Geq}$.

    Let now $\hat{\mathfrak{G}}$ be a degree-$O(d)$ derivation of $\sigma(f)$ from $\mathcal{F}_{\Gamma}$. It is also possible to find a degree-$O(d)$ derivation of $\sigma(f)$ from $\mathcal{F}_{\Geq}$: since the domain axioms $f_{\Lambda}(\sigma(z))$ are already lines of $\mathcal{F}_{\Geq}$, it suffices to derive the polynomial $f_C - \sigma(f_C)$ and subtract it from the axiom $f_C$, thus obtaining the set given by the $\sigma(f_C)$s, i.e., the starting points of the derivation from $\mathcal{F}_{\Gamma}$. Finally, adding the derivation of $f - \sigma(f)$ from $\mathcal{F}_{\Geq}$ we obtain the claimed $\PC$ derivation of $f$.
\end{proof}

It remains to transform a bounded-degree $\PC$ derivation of $f$ from $\mathcal{F}_{\Geq}$ into one from $\mathcal{F}_\Delta$, in polynomial time. Consider a constraint $C = R_C(x_{m+1}, \ldots, x_{m+k}) \in \mathcal{C}_{\Delta}$ with pp-formula
\begin{align}
    R_C(x_{m+1}, \ldots, x_{m+k}) = \exists x_1 \ldots x_m\, L_1^C \ldots L_s^C,
\end{align}
the $L_j^C$ being constraints from $\Geq$ on $\{x_1, \ldots, x_{m+k}\}$. For every $a \in R_C \subseteq \Lambda^k$ there exists at least one witness $\bar{a} = (a_1, \ldots, a_m) \in \Lambda^m$ so that $(\bar{a},a)$ satisfies every $L_j^C$; fixing a \emph{witness map} $w_C \colon R_C \to \Lambda^m$ choosing one such $\bar{a}$ per $a$,
\begin{align}
    &f_{L_j^C}(w_C(a),a) = 0 \qquad \forall j \in [s] \label{eq:interpolating constraint of pp-formula equals 0}
\end{align}
for every $a \in R_C$. Let $U_i \in \mathbb{Q}[x_{m+1}, \ldots, x_{m+k}]$ interpolate the $i$-th witness coordinate, $U_i(a) = w_C(a)_i$; these are well-defined since $w_C$ is a function. With $U = (U_1, \ldots, U_m)$, \cref{eq:interpolating constraint of pp-formula equals 0} becomes
\begin{align}
    &f_{L_j^C}(U(a),a) = 0\qquad \forall j \in [s], \; \forall a \in R_C\label{eq:constraint polynomials with interpolated variables}.
\end{align}
Writing $\mathcal{F}_C := \{f_C(x_{m+1}, \ldots, x_{m+k})\} \cup \{f_{\Lambda}(x)\}_{x \in \{x_{m+1}, \ldots, x_{m+k}\}}$, the interpolants let us pull back the axioms of $\mathcal{F}_{\Geq}$ into $\mathcal{F}_\Delta$. First, by applying the Strong Nullstellensatz, one can verify the following congruences:

\begin{lemma}\label{th:generators with interpolated variables are in the ideal}
    \begin{enumerate}
        \item $f_{\Lambda}(U_i(x_{m+1}, \ldots, x_{m+k})) \cong 0 \mod \langle \mathcal{F}_C \rangle$ for all $i \in [m]$,
        \item $f_{L_j^C}(U_1, \ldots, U_m, x_{m+1}, \ldots, x_{m+k}) \cong 0 \mod \langle \mathcal{F}_C  \rangle$ for all $j \in [s]$.
    \end{enumerate}
\end{lemma}

\begin{proof}
    By \cref{th:domain ideal nullstellensatz} the ideal $\langle \mathcal{F}_C \rangle$ is radical. Therefore, by applying Hilbert Strong Nullstellensatz, to show whether a target polynomial belongs to $\langle \mathcal{F}_C \rangle$, it suffices to check whether for all points in $\mathbf{V}(\mathcal{F}_C)$ the target polynomial equals $0$. We also recall that $\mathbf{V}(\mathcal{F}_C) = R_C$.

    Let $a \in R_C$ be a point in the variety. For every interpolant polynomial it holds that $U_i(a) \in \Lambda$, thus $f_{\Lambda}(U_i(a)) = 0$, and that \cref{eq:interpolating constraint of pp-formula equals 0}, proving the first item. The second is proven similarly, using the definition of $U_i$ and of $f_{L_j^C}$.
\end{proof}

\Cref{th:generators with interpolated variables are in the ideal} only places the pulled-back axioms in $\langle \mathcal{F}_C \rangle$. For a $\PC$ reduction we need more: by completeness of $\PC$ over finite domains, each pulled-back axiom admits a \emph{constant-degree} derivation from $\mathcal{F}_\Delta$.

\begin{proposition}\label{th:PC size bounds for "generators with interpolated variables are in the ideal"}
    \begin{enumerate}
        \item For all $i \in [m]$ a degree-$O(1)$ derivation of $f_{\Lambda}(U_i)$ from $\mathcal{F}_C$ is computable in polynomial time.
        \item For all $j \in [s]$ a degree-$O(1)$ derivation of $f_{L_j^C}(U_1, \ldots, U_m, x_{m+1}, \ldots, x_{m+k})$ from $\mathcal{F}_C$ is computable in polynomial time.
    \end{enumerate}
\end{proposition}

\begin{proof}
    It suffices to observe that the claims in \cref{th:generators with interpolated variables are in the ideal} do not depend on the number of variables $n$, since $\Delta$ and $\Gamma$ are finite. The systems involved therefore have $O(1)$ variables and $O(1)$ degree, and the argument of \cref{sect:preliminaries} applies unchanged, so the derivations have degree $O(1)$ and are computable in time polynomial in $n$.
\end{proof}

The derivation to be pulled back is the one provided by \cref{th:equality collapse}. Substituting constraint by constraint these bounded-degree pullbacks into a derivation of $f$ from $\mathcal{F}_{\Geq}$ turns this into a derivation from $\mathcal{F}_\Delta$.

\begin{theorem}[$\PC$ reduction by pp-definability]\label{th:pp-definability and PC reductions}
    Let $\Gamma,\Delta$ be two fixed constraint languages over the finite domain $\Lambda$, with $\Gamma$ pp-defining $\Delta$. Then $\PC$-$\IMP_d(\Delta)$ $\PC$-reduces to $\PC$-$\IMP_d(\Gamma)$. In particular, if $\PC$-$\IMP_{d}(\Gamma)$ can be solved in polynomial time for any fixed $d$, then so can $\PC$-$\IMP_d(\Delta)$.
\end{theorem}

\begin{proof}
    Let $(f, \mathcal{T}_{\Delta})$ be an instance of $\PC$-$\IMP_d(\Delta)$, where $\mathcal{T}_{\Delta} = (X, \Lambda, \mathcal{C}_{\Delta})$ is an instance of $\CSP(\Delta)$. Consider $\mathcal{T}_{\Geq}$ defined as in \cref{sect:pp-definitions}--for which \cref{th:pp-definability IMP reduction} holds--and, furthermore, consider the sets of generators $\mathcal{F}_{\Delta}$ and $\mathcal{F}_{\Geq}$ defined in \cref{eq:generators Delta} and \cref{eq:generators Gamma}, respectively. 
    
    Let $\mathfrak{G} = [g_1, g_2, \ldots, g_s = f]$ be any degree-$O(d)$ $\PC$-derivation of $f$ from $\mathcal{F}_{\Geq}$; by \cref{th:equality collapse} one is computable from a degree-$O(d)$ derivation of $\sigma(f)$ from $\mathcal{F}_{\Gamma}$ in time polynomial in $n$ and in the bit-size of the latter. We will define a derivation $\mathfrak{F} = [f_1, f_2, \ldots, f_r = f]$ simulating $\mathfrak{G}$ with only a polynomial increase in size. We will proceed by induction. The idea is as follows: consider any polynomial $g \in \mathfrak{G}$. Then $g \in \mathfrak{G}$ is a polynomial in the variables $Y \cup X$, where $Y$ is the set of variables introduced when pp-defining the constraints $\mathcal{C}_{\Delta}$ using $\Gamma$. For each such $g$ we define $f_g(X) := g(U_Y(X),X)$, where $U_Y$ is the set of interpolating polynomials. If we can show that each $f_g$ can be $\PC$ derived from $\mathcal{F}_{\Delta}$ with bounded degree and polynomial size, then we are done: for the final element $g_s = f$, which already contains only variables in $X$, $f_{g_s}(X) = g_s(X) = f$, so $f_r = f$.
    
    \textit{Base case--Domain polynomials.} There are two types of domain polynomials in the instance $(f, \mathcal{T}_{\Geq})$: either $g = f_{\Lambda}(x)$ for some $x \in X$, in which case trivially $g \in \mathcal{F}_{\Delta}$ and we can add $g$ to $\mathfrak{F}$; otherwise $g = f_{\Lambda}(y)$ for some variable $y \notin X$ and arising from some pp-formula of a constraint $C = R_C(x_{m+1}, \ldots, x_{m+k}) \in \mathcal{C}_{\Delta}$ from $\Gamma$. In the latter case, there exists an interpolant polynomial $U_y$. \cref{th:PC size bounds for "generators with interpolated variables are in the ideal"} guarantees that $f_g := f_{\Lambda}(U_y)$ has a bounded-degree $\PC$ derivation from $\mathcal{F}_C \subseteq \mathcal{F}_{\Delta}$ of polynomial size. Hence, we add $f_g$ and its $\PC$ derivation from $\mathcal{F}_{\Delta}$ to $\mathfrak{F}$.

    \textit{Base case--Constraint polynomials.} These are polynomials also arising from the pp-definition of $\Delta$ by $\Gamma$ of the form $g = f_{L_j^C}(y_1, \ldots, y_m, x_{m+1}, \ldots, x_{m+k})$, equalities included. Similarly as before, \cref{th:PC size bounds for "generators with interpolated variables are in the ideal"} guarantees that $f_g := f_{L_j^C}(U_1, \ldots, U_m, x_{m+1}, \ldots, x_{m+k})$ has a polynomial-size bounded-degree $\PC$ derivation from $\mathcal{F}_{\Delta}$. We add $f_g$ and its derivation to $\mathfrak{F}$.

    \textit{Inductive step--Addition.} Suppose $g = g_1 + g_2$ and that $f_{g_1} = g_1(U_Y(X),X)$ and $f_{g_2} = g_2(U_Y(X),X)$ can be derived from $\mathcal{F}_{\Delta}$. Then $f_g := f_{g_1} + f_{g_2}$ can also be derived from $\mathcal{F}_{\Delta}$. We simply add $f_g$ to $\mathfrak{F}$.

    \textit{Inductive step--Multiplication.} Suppose $g = z g_1 \in \mathfrak{G}$ with $f_{g_1} = g_1(U_Y(X), X)$ in $\mathfrak{F}$. If $z = x \in X$, then $f_g := x f_{g_1}$ is one $\PC$ multiplication step and we add it to $\mathfrak{F}$. If $z = y \in Y$, then $f_g := U_y f_{g_1}$ has a straightforward $\PC$ derivation of bounded degree and polynomial size, which we add to $\mathfrak{F}$.

    In the terms of \cref{def:PC reduction}, $r$ fixes $X$ and sends each $y \in Y$ to $U_y$, the base cases are the derivations of the $r(h)$, and $f - r(f) = 0$ since $f$ involves only $X$. If $\PC$-$\IMP_d(\Gamma)$ is solvable, then a degree-$O(d)$ derivation of $\sigma(f)$ from $\mathcal{F}_{\Gamma}$ is computable in time $n^{O(d)}$, hence of bit-size $n^{O(d)}$, and the simulation above turns it into one of $f$ from $\mathcal{F}_{\Delta}$ within the same bound.
\end{proof}

\subsection{Pp-interpretations}\label{sec:pp-interpretations}

Reductions by pp-definability are useful, yet they require both languages to be over the same domain. In this section, we show how pp-interpretations allow for reductions of $\PC$-$\IMP$ between languages over different domains.

\begin{definition}\label{def:pp-interpretability}
    Let $\Gamma, \Delta$ be fixed constraint languages over finite domains $\Lambda_{\Gamma}, \Lambda_{\Delta}$, respectively. We say that $\Gamma$ \emph{pp-interprets} $\Delta$ if there exists a natural number $\ell$, a set $F \subseteq \Lambda_{\Gamma}^{\ell}$, and an onto mapping $\pi:F \rightarrow \Lambda_{\Delta}$ such that $\Gamma$ pp-defines the following relations
    \begin{enumerate}
        \item the relation $F$,
        \item the $\pi$-preimage of the equality relation on $\Lambda_{\Delta}$,
        \item the $\pi$-preimage of every relation in $\Delta$,
    \end{enumerate}
    where by the $\pi$-preimage of a $k$-ary relation $S$ on $\Lambda_{\Delta}$ we mean the $\ell k$-ary relation $\pi^{-1}(S)$ on $\Lambda_{\Gamma}$ defined by
    \begin{align*}
        \pi^{-1}(S)(y_{11},\ldots, y_{\ell 1}, y_{1 2}, \ldots, y_{\ell 2}, \ldots, y_{1 k}, \ldots, y_{\ell k}) \quad is \ true
    \end{align*}
    if and only if
    \begin{align*}
        S(\pi(y_{11}, \ldots, y_{\ell 1}), \ldots, \pi(y_{1k}, \ldots, y_{\ell k})) \quad is \ true.
    \end{align*}
\end{definition}

\begin{definition}\label{def:pp-encodings}
    If $\pi$ is a bijection we say that $\Gamma$ \emph{pp-encodes} $\Delta$ (or $\Delta$ is \emph{pp-encoded} by $\Gamma$).
\end{definition}

Consider $\Gamma,\Delta$ constraint languages over finite domains $\Lambda_{\Gamma}, \Lambda_{\Delta}$, respectively, such that $\Delta$ is pp-interpreted by $\Gamma$ via $\pi : F \rightarrow \Lambda_{\Delta}$ with $F \subseteq \Lambda_{\Gamma}^{\ell}$. Since $\pi$ is onto, we may fix once and for all a \emph{section} of $\pi$, that is, a map $s \colon \Lambda_{\Delta} \to F$ with $\pi(s(e)) = e$ for every $e \in \Lambda_{\Delta}$; note that $s$ is injective, and that when $\pi$ is a bijection (the pp-encoding case of \cref{def:pp-encodings}) the only section is $s = \pi^{-1}$.

Let $(f_{\Delta}, \mathcal{T}_{\Delta})$ be an instance of $\PC$-$\IMP_d(\Delta)$ with $\mathcal{T}_{\Delta} = \{X = \{x_1, \ldots, x_n\}, \Lambda_{\Delta}, \mathcal{C}_{\Delta}\}$. Our objective is to show a polynomial time reduction of $\PC$-$\IMP_d(\Delta)$ to $\PC$-$\IMP_c(\Gamma)$ for some constant natural number $c$.

First we define a language $\Gamma'$ and instance $(f_{\Gamma'}, \mathcal{T}_{\Gamma'})$ of $\IMP_{d \ell |\Lambda_{\Gamma}|}(\Gamma')$, where $\mathcal{T}_{\Gamma'} = (Y, \Lambda_{\Gamma}, \mathcal{C}_{\Gamma'})$ and $Y = \{y_{1 1}, \ldots, y_{\ell 1}, \ldots, y_{1 n},\ldots, y_{\ell n}\}$, that is a polynomial-time reduction of $(f_{\Delta}, \mathcal{T}_{\Delta})$ with respect to the $\IMP$, with the key property that $\Gamma$ pp-defines $\Gamma'$.
To do so, we define $\Gamma' := \{F\} \cup \{\pi^{-1}(S) : S \in \Delta\}$.
Then $\Gamma$ pp-defines $\Gamma'$ by \cref{def:pp-interpretability}.
Coherently, for every $k$-ary constraint $S(x_{i_1}, \ldots, x_{i_k}) \in \mathcal{C}_{\Delta}$, we add the constraint $\pi^{-1}(S)(y_{1 i_1}, \ldots, y_{\ell i_1}, \ldots, y_{1 i_k}, \ldots, y_{\ell i_k})$ to $\mathcal{C}_{\Gamma '}$; moreover, for \emph{every} $i \in [n]$ we add the constraint $F(y_{1 i}, y_{2 i}, \ldots, y_{\ell i})$ to $\mathcal{C}_{\Gamma'}$.
The $F$-constraints guarantee that every tuple $(y_{1i}, \ldots, y_{\ell i})$ assumes only values in $F$, so that $\pi(y_{1i}, \ldots, y_{\ell i})$ corresponds to a legal assignment of $x_i$ in the domain $\Lambda_{\Delta}$. For a variable $x_i$ appearing in some constraint of $\mathcal{C}_{\Delta}$, the corresponding $F$-constraint is implied by the $\pi^{-1}(S)$-constraints (every $\ell$-block of a tuple of $\pi^{-1}(S)$ lies in $F$, by \cref{def:pp-interpretability}), so adding it changes neither the solution set nor, by radicality, the ideal; for a variable appearing in no constraint, the $F$-constraint is necessary: without it, a solution of $\mathcal{T}_{\Gamma'}$ may assign $(y_{1i},\ldots,y_{\ell i})$ a tuple outside $F$, which would not correspond to a feasible assignment of $x_i$ within the domain $\Lambda_{\Delta}$.

Next, let $p$ be the polynomial interpolating $\pi$ (of total degree at most $\ell|\Lambda_{\Gamma}|$); note that the $F$-constraints guarantee that nothing depends on the values of $p$ outside $F$. We define $f_{\Gamma'}$ to be the polynomial obtained from $f_{\Delta}$ by substituting every appearance of variable $x_i$ with the polynomial $p(y_{1 i}, \ldots, y_{\ell i})$. Let
\begin{align}\label{eq:generators pp-interpretations case}
    \mathcal{F}_\Delta := \{ \bigcup_{C \in \mathcal{C}_{\Delta}} f_C \} \cup \{f_{\Lambda_{\Delta}}(x)\}_{x \in X} \quad \text{and} \quad \mathcal{F}_{\Gamma'} := \{ \bigcup_{C \in \mathcal{C}_{\Gamma'}}f_C\} \cup \{ f_{\Lambda_{\Gamma}}(y)\}_{y \in Y}
\end{align}
be the two sets of generators. The following lemma justifies such reduction.

\begin{lemma}[\cite{BulatovRSTOC22}]\label{th:pp-interpretability IMP reduction}
    $f_{\Delta} \in \langle \mathcal{F}_{\Delta} \rangle$ if and only if $f_{\Gamma'} \in \langle \mathcal{F}_{\Gamma'} \rangle$. Moreover, because $\Gamma$ pp-defines $\Gamma'$, $\IMP_d(\Delta)$ is polynomial time reducible to $\IMP_{d \ell |\Lambda_{\Gamma}|}(\Gamma)$.
\end{lemma}

As in \cref{sect:pp-definitions}, the construction above supplies the polynomial-time map on instances, and the rest of this subsection shows that this is indeed a $\PC$ reduction.

We wish to show that a $\PC$ derivation of $f_{\Gamma'}$ from $\mathcal{F}_{\Gamma'}$ can be efficiently simulated obtaining a $\PC$ derivation of $f_{\Delta}$ from $\mathcal{F}_{\Delta}$. Note that there is a correspondence between constraint polynomials in $\mathcal{F}_{\Delta}$ and constraint polynomials in $\mathcal{F}_{\Gamma'}$, namely for each $f_{C_S}(X) \in \mathcal{F}_{\Delta}$ with $S \in \Delta$, there exists $f_{C_{\pi^{-1}(S)}}(Y) \in \mathcal{F}_{\Gamma'}$.

We introduce polynomial interpolants mimicking the section $s$. For every $j \in [\ell]$ and $i \in [n]$, let $U_{j,i} \in \mathbb{Q}[x_i]$ be the polynomial of degree at most $|\Lambda_{\Delta}|-1$ interpolating the $j$-th coordinate of the section, that is,
\begin{align}\label{eq:section interpolants}
    U_{j,i}(e) = s(e)_j \qquad \text{for every } e \in \Lambda_{\Delta}.
\end{align}
The fact that $s$ is a function guarantees that the $U_{j,i}$ are well-defined polynomial functions (in the pp-encoding case, $s = \pi^{-1}$ and we recover the previous definition).
We set $U_i(x_i) = (U_{1,i}(x_i), \ldots, U_{\ell,i}(x_i))$, so that $U_i(e) = s(e) \in F$ for every $e \in \Lambda_{\Delta}$, and then set $U(X) = (U_1(x_1), \ldots, U_{n}(x_n))$.

\begin{lemma}[see also \cite{BortolottiMV_ARXIV25}]\label{th:variables interpolation and domain polynomials cong 0 for pp-interpretations}
    Let $\mathcal{F}_{C_S} : = f_{C_S}(X_{C_S}) \cup \{f_{\Lambda_{\Delta}}(x)\}_{x \in X_{C_S}}$ with $X_{C_S} := Var(C_S)$. The following congruences hold.
    \begin{enumerate}
        \item $x_i - p(U_i(x_i)) \cong 0 \mod{\langle f_{\Lambda_{\Delta}}(x_i) \rangle}$ for all $i \in [n]$,
        \item $f_{\Lambda_{\Gamma}}(U_{j,i}(x_i)) \cong 0 \mod{\langle f_{\Lambda_{\Delta}}(x_i)\rangle}$ for all $j \in [\ell]$, $i \in [n]$,
        \item $f_{C_F}(U_i(x_i)) \cong 0 \mod{\langle f_{\Lambda_{\Delta}}(x_i) \rangle}$ for all $i \in [n]$,
        \item $f_{C_{\pi^{-1}(S)}}(U(X_{C_S})) \cong 0 \mod \langle \mathcal{F}_{C_S} \rangle$ for all constraints $C_S \in \mathcal{C}_{\Delta}$.
    \end{enumerate}
    Furthermore, all the above polynomials on the LHS can be derived from $\mathcal{F}_{\Delta}$ by constant-degree $\PC$ in polynomial time.
\end{lemma}

\begin{proof}
    By \cref{th:domain ideal nullstellensatz} the ideals $\langle f_{\Lambda_{\Delta}}(x_i) \rangle$ and $\langle \mathcal{F}_{C_S} \rangle$ are radical, so by Hilbert's Strong Nullstellensatz it suffices to check that each polynomial on the LHS vanishes on every point of the corresponding variety.
    For items 1--3, the variety is $\Lambda_{\Delta}$, so let $e \in \Lambda_{\Delta}$: then $p(U_i(e)) = p(s(e)) = \pi(s(e)) = e$, proving item 1; $U_{j,i}(e) = s(e)_j \in \Lambda_{\Gamma}$, proving item 2; and $U_i(e) = s(e) \in F$, proving item 3.
    For item 4, the variety is $\mathbf{V}(\mathcal{F}_{C_S}) = S$, so let $a = (a_1, \ldots, a_k) \in S$: then
    \[
        \big(U_{i_1}(a_1), \ldots, U_{i_k}(a_k)\big) = \big(s(a_1), \ldots, s(a_k)\big) \in \pi^{-1}(S),
    \]
    since $\big(\pi(s(a_1)), \ldots, \pi(s(a_k))\big) = (a_1, \ldots, a_k) \in S$; hence $f_{C_{\pi^{-1}(S)}}$ vanishes at this point.

    All polynomials on the LHS depend on a constant number of variables, thus their $\PC$ derivation from $\mathcal{F}_{\Delta}$ can be performed with constant degree and polynomial time (in the same manner as in \cref{sect:encodings}).
\end{proof}

Next, we show how to efficiently simulate a $\PC$ derivation from $\mathcal{F}_{\Gamma'}$.

\begin{lemma}\label{th:simulation pp-interpretations}
    Let $\mathfrak{G} = [g_1, g_2, \ldots, g_t]$ be a degree-$c$ $\PC$ derivation from $\mathcal{F}_{\Gamma'}$ of size $s$. Let $f_i(X) := g_i(U(X))$. For fixed $c$, a degree-$O(c)$ $\PC$ derivation of $f_i(X)$ from $\mathcal{F}_{\Delta}$ is computable in time polynomial in $s$.
\end{lemma}

\begin{proof}
    The argument is the same as in the proof of \cref{th:pp-definability and PC reductions}. The base case, namely when $g_i \in \mathcal{F}_{\Gamma'}$, is covered by \cref{th:variables interpolation and domain polynomials cong 0 for pp-interpretations} (the $F$-constraint polynomials by item 3, the $\pi^{-1}(S)$-constraint polynomials by item 4, the domain polynomials by item 2). The inductive case is straightforward. In particular, if $g_t = y_{ji}g_{t-1}$ for some $y_{ji} \in Y$, then the polynomial $f_t(X) = U_{j,i}(x_i)f_{t-1}(X)$ can be derived efficiently from the polynomial $f_{t-1}(X)$.
\end{proof}

\begin{theorem}[$\PC$ reduction by pp-interpretations]\label{th:pp-interpretations and PC reductions}
    Let $\Gamma,\Delta$ be two fixed constraint languages over finite domains $\Lambda_{\Gamma},\Lambda_{\Delta}$, respectively. Suppose that $\Gamma$ pp-interprets $\Delta$ with $\ell$. Then $\PC$-$\IMP_d(\Delta)$ $\PC$-reduces to $\PC$-$\IMP_{d\,\ell\,|\Lambda_{\Gamma}|}(\Gamma)$. In particular, if $\PC$-$\IMP_c(\Gamma)$ can be solved in polynomial time for any fixed $c$, then so can $\PC$-$\IMP_d(\Delta)$ for any fixed $d$.
\end{theorem}

\begin{proof}
    Let $(f_{\Delta}, \mathcal{T}_{\Delta})$ be an instance of $\PC$-$\IMP_d(\Delta)$ and consider the instance $(f_{\Gamma'}, \mathcal{T}_{\Gamma'})$ of $\PC$-$\IMP_{d\ell|\Lambda_{\Gamma}|}(\Gamma')$ defined as above. By construction, $\Gamma$ pp-defines $\Gamma'$, so $\PC$-$\IMP_c(\Gamma')$ $\PC$-reduces to $\PC$-$\IMP_c(\Gamma)$ by \cref{th:pp-definability and PC reductions}; moreover, by \cref{th:pp-interpretability IMP reduction}, $f_{\Delta} \in \langle \mathcal{F}_{\Delta} \rangle$ if and only if $f_{\Gamma'} \in \langle \mathcal{F}_{\Gamma'} \rangle$.
    Given a polynomial-size $\PC$ derivation of $f_{\Gamma'}$ from $\mathcal{F}_{\Gamma'}$ of degree $c$, \cref{th:simulation pp-interpretations} computes in polynomial time a bounded-degree $\PC$ derivation of $f_{\Gamma'}(U(X))$ from $\mathcal{F}_{\Delta}$. By definition, $f_{\Gamma'}(U(X)) = f_{\Delta}\big(p(U_1(x_1)), \ldots, p(U_n(x_n))\big)$. Recovering $f_{\Delta}(X)$ from it is handled as the equalities were in \cref{th:renaming}, applied to the substitution of each $x_i$ by $p(U_i(x_i))$: its argument holds verbatim when a variable is replaced by a polynomial, and its assumption is met by the derivation of $x_i - p(U_i(x_i))$ from $\mathcal{F}_{\Delta}$ supplied by \cref{th:variables interpolation and domain polynomials cong 0 for pp-interpretations}. We thus obtain a derivation of $f_{\Delta}$ from $\mathcal{F}_{\Delta}$ of degree $O(d)$ and polynomial size, which we output. In particular, if $\PC$-$\IMP_c(\Gamma)$ is solvable, then so is $\PC$-$\IMP_c(\Gamma')$, and the derivation of $f_{\Gamma'}$, hence the one of $f_{\Delta}$, is computable in polynomial time.
\end{proof}

\begin{corollary}[$\PC$ reduction by pp-encodings]\label{th:pp-encodings and PC reductions}
    \Cref{th:pp-interpretations and PC reductions} holds, in particular, when $\Gamma$ pp-encodes $\Delta$.
\end{corollary}

\runin{Reducing the semilattice.}
We illustrate the usefulness of pp-interpretations as a tool for transferring $\PC$-$\IMP$ tractability across domains. Here we reduce $\PC$-$\IMP_d(\Delta)$, with $\Delta$ a constraint language on a finite domain $\Lambda_{\Delta}$ closed under a semilattice polymorphism, to $\PC$-$\IMP_c(\Gamma)$ with $\Gamma$ Boolean and closed under the \textsc{Min} (or \textsc{Max}) operation. This instance happens to be a pp-encoding, since the map $\pi$ below is a bijection; it is thus also a witness to the usefulness of the encoding special case.

Recall that a binary operation $\phi:\Lambda^2 \rightarrow \Lambda$ is a \emph{semilattice} operation if it satisfies the properties: \textit{Associativity}: $\phi(x,\phi(y,z)) = \phi(\phi(x,y),z)$; \textit{Commutativity}: $\phi(x,y) = \phi(y,x)$; \textit{Idempotency}: $\phi(x,x) = x$. Setting $a \leq_{\phi} b$ whenever $\phi(a,b) = a$ defines a partial order on $\Lambda$ with respect to which $\phi$ is the \emph{meet}, that is, $\phi(x,y) = \max \{ z \ | \ z \leq_{\phi} x, \ z \leq_{\phi} y\}$.

\begin{example}[Reducing a semilattice to the Boolean case]\label{ex:semilattice}
    Let $\Delta$ be a constraint language over a finite domain $\Lambda_{\Delta} = \{0, 1, \ldots, \ell\}$ closed under a meet-semilattice polymorphism $\phi$; we use a pp-encoding to reduce $\PC$-$\IMP_d(\Delta)$ to a $\PC$-$\IMP$ over a Boolean \textsc{Min}-closed language, which is tractable.

    Take
    \begin{align*}
        \tilde{\pi} : \Lambda_{\Delta} \to \{0,1\}^{\ell + 1}, \qquad \tilde{\pi}(a)_t = \mathbf{1}[ \, t \leq_{\phi} a \,] \quad \forall t \in \Lambda_{\Delta},
    \end{align*}
    where $\mathbf{1}[\cdot]$ is an indicator, and let $F := \tilde{\pi}(\Lambda_{\Delta})$ be the set of $0/1$ encodings. The map $\tilde{\pi}$ is a bijection onto $F$, and we let $\pi := \tilde{\pi}^{-1} : F \rightarrow \Lambda_{\Delta}$. Observing that $t \leq_{\phi} \phi(a,b)$ if and only if $t \leq_{\phi} a$ and $t \leq_{\phi} b$, we get $\tilde{\pi}(\phi(a,b)) = \textsc{Min}(\tilde{\pi}(a), \tilde{\pi}(b))$ and $\tilde{\pi}(S)$ is \textsc{Min}-closed for every $\phi$-closed $S$. Starting from $\Gamma = \emptyset$, add to $\Gamma$ the relation $F$, the $\pi$-preimage of the equality relation on $\Lambda_{\Delta}$, and the $\pi$-preimage of every relation in $\Delta$. Then $\Gamma$ is \textsc{Min}-closed and pp-encodes $\Delta$. Since $\PC$-$\IMP_d(\Gamma)$ is solvable in polynomial time by bounded-degree $\PC$ for any Boolean \textsc{Min}/\textsc{Max}-closed $\Gamma$ \cite{BortolottiMV_ARXIV25}, \cref{th:pp-encodings and PC reductions} yields that $\PC$-$\IMP_d(\Delta)$ is solvable in polynomial time by bounded-degree $\PC$ for every fixed $d$.

    To see how a $\PC$ derivation is transported, take $\Lambda_{\Delta} = \{0,1,2\}$ with $\phi = \min$, so that $\ell = 2$, $F = \{(1,0,0), (1,1,0), (1,1,1)\}$ and $\pi(u) = \sum_{t} u_t - 1$. Thus \eqref{eq:section interpolants} gives
    \begin{align*}
        U_{1,i}(x_i) = 1, \qquad U_{2,i}(x_i) = \frac{x_i(3 - x_i)}{2}, \qquad U_{3,i}(x_i) = \frac{x_i(x_i - 1)}{2},
    \end{align*}
    whence $U_{1,i}(x_i) + U_{2,i}(x_i) + U_{3,i}(x_i) - 1 = x_i$, as expected from item 1 of \cref{th:variables interpolation and domain polynomials cong 0 for pp-interpretations}. Multiplying a line $f$ of the Boolean derivation by the variable encoding the second coordinate of $x_i$ pulls back to $U_{2,i}(x_i) f$, which $\PC$ directly obtains as $\tfrac{3}{2}(x_i f) - \tfrac{1}{2}\big(x_i(x_i f)\big)$. Note that the degree increase is $\deg U_{2,i} = |\Lambda_{\Delta}| - 1 = 2$.
\end{example}

While this result was already pointed out in \cite{BortolottiMV25,BortolottiMV_ARXIV25}, pp-encodings provide a general framework for doing such reductions.

\section{$\CSP$-based Complexity of $\PC$-$\IMP$: the Boolean case}\label{sect:csp-pc-imp}

An interesting consequence of \cref{th:pp-definability and PC reductions} is that the algebraic approach to $\CSP$s can be leveraged to give a constraint-language based complexity classification of the $\PC$-$\IMP_d$ problem.
    This follows from the tight relation between pp-definability and polymorphisms, where we recall that $\Pol(\Gamma)$ is the set of operations preserving every relation of $\Gamma$ (see e.g. \cite{BartoKW17}).

\begin{lemma}\label{old:th:pp-definability and polymorphism clone}
    Let $\Gamma$ and $\Delta$ be constraint languages over the same domain $D$. Then $\Gamma$ pp-defines $\Delta$ if and only if $\Pol(\Gamma) \subseteq \Pol(\Delta)$.
\end{lemma}

\subsection{The Boolean dichotomy for $\IMP_d(\Gamma)$.}

In this section we do a step back and consider the decision problem $\IMP_d(\Gamma)$. The reason is that solvability of $\PC$-$\IMP_d(\Gamma)$ places $\IMP_d(\Gamma)$ in $\Ptime$; conversely, hardness for $\IMP_d(\Gamma)$ rules it out. Mastrolilli \cite{Mastrolilli21} applied the algebraic approach to $\IMP_d$ and delineated its tractability borderline for Boolean constraint languages.

\begin{theorem}[\cite{Mastrolilli21}]\label{th:Mastrolilli Boolean IMP}
    Let $\Gamma$ be a finite Boolean constraint language. If the solution space of every relation in $\Gamma$ is closed under any of \{\textsc{Majority}, \textsc{Minority}, \textsc{Max}, \textsc{Min}\}, then $\IMP_d(\Gamma)$ can be solved in $n^{O(d)}$ time for any $d \geq 0$. Otherwise, there is a constant $d \in \{0, 1, 2\}$ such that $\IMP_d(\Gamma)$ is \coNP-complete.
\end{theorem}

We emphasize that these algorithms do not work within bounded-degree $\PC/\mathbb{Q}$; in particular the \textsc{Minority} algorithm relies on the \lex order and does not guarantee bounded degree \cite{BharathiM25}.

We sharpen \cref{th:Mastrolilli Boolean IMP} by closing the only case left open. Consider the three Boolean unary operations: the constants $\mathbf{0}\colon x \mapsto 0$ and $\mathbf{1}\colon x \mapsto 1$, and the negation $\neg\colon x \mapsto 1-x$. Mastrolilli proved $\IMP_1(\Gamma)$ to be \coNP-complete for $\Gamma$ closed under $\mathbf{0}$, $\mathbf{1}$, $\neg$ but under none of \textsc{Majority}, \textsc{Minority}, \textsc{Min}, \textsc{Max}, and conjectured the same when $\neg \notin \Pol(\Gamma)$.

Let $\pin{0} = \{0\}$ and $\pin{1} = \{1\}$ be the unary relations over the Boolean domain. We recall that repeated variables in a constraint are allowed.

\begin{lemma}\label{lem:orientation}
    Suppose that $\mathbf{0}, \mathbf{1} \in \Pol(\Gamma)$ but that $\neg \notin \Pol(\Gamma)$.
    Then there exist a fixed relation $R \in \Gamma$ of arity $k$ and a fixed tuple $a \in R$ with the following property. For two variables $u,v$, let $G_a(u,v)$ denote the constraint
    \begin{equation}\label{eq:guard}
        G_a(u,v) = R(y_1, \ldots, y_k), \qquad y_i = \begin{cases} u, & a_i = 0,\\ v, & a_i = 1,\end{cases} \qquad i \in [k],
    \end{equation}
    obtained by placing $u$ in every coordinate where $a$ is $0$ and $v$ in every coordinate where $a$ is $1$. Then $G_a(u,v)$ is satisfied exactly by the pairs $(0,0), (0,1), (1,1)$, that is, precisely when $u \leq v$.
\end{lemma}

\begin{proof}
    Since negation does not preserve $\Gamma$, choose $R \in \Gamma$ and $a \in R$ with $\bar a \notin R$, where $\bar a_i = 1 - a_i$. The values taken by $(y_1, \ldots, y_k)$ at $(u,v) = (0,0),(0,1),(1,0),(1,1)$ are respectively $0^k, a, \bar a, 1^k$. The first, second and fourth belong to $R$, and the third does not. This proves the claim.
\end{proof}

\begin{remark}
    The relation of $G_a(u,v)$ is $R \in \Gamma$, so that $G_a(u,v)$ is a constraint over $\Gamma$. It uses only the two variables $u,v$, each occurring several times.
\end{remark}

\begin{proposition}\label{prop:degree-one-hardness}
    Let $\Gamma$ be a finite Boolean constraint language. If the solution space of every relation in $\Gamma$ is closed under $\mathbf{0}$ and $\mathbf{1}$, but under none of $\neg$, \textsc{Majority}, \textsc{Minority}, \textsc{Max}, \textsc{Min}, then $\IMP_1(\Gamma)$ is \coNP-complete.
\end{proposition}

\begin{proof}
    To begin observe that $\CSP(\Gpin)$ is \NP-complete, where $\Gpin := \Gamma \cup \{\pin{0}, \pin{1}\}$. Indeed, $\Pol(\Gpin) \subseteq \Pol(\Gamma)$, thus no new polymorphism is introduced, and the constant operations do not preserve the two pins. None of Schaefer's six tractable operations preserves $\Gpin$, hence $\CSP(\Gpin)$ is \NP-complete \cite{Jeavons98,Schaefer78}.

    We now reduce $\CSP(\Gpin)$ to the complement of $\IMP_1(\Gamma)$ in polynomial time. Let $\mathcal{T}_{\Gpin} = (X, \BoolD, \mathcal{C})$ be an instance of $\CSP(\Gpin)$, and let $P_0$ and $P_1$ be the sets of variables pinned to $0$ or to $1$, respectively, namely $P_0 = \{x \in X \mid \pin{0}(x) \in \mathcal{C}\}$ and similarly for $P_1$. Let moreover $R$, $a$ and $G_a$ be as in \cref{lem:orientation}.

    Suppose first that $P_0 \cap P_1 \neq \emptyset$. Some variable is then pinned both to $0$ and to $1$, so that $\mathcal{T}_{\Gpin}$ has no solution, and the reduction must produce a satisfiable instance of $\IMP_1(\Gamma)$. We take
    \[
        \mathcal{T} := (\{z_0, z_1\}, \BoolD, \{G_a(z_0,z_1), G_a(z_1,z_0)\}), \qquad f := z_1 - z_0.
    \]
    Let $(b_1, b_2)$ be a solution of $\mathcal{T}$. By \cref{lem:orientation} the two constraints of $\mathcal{T}$ force $b_1 \leq b_2$ and $b_2 \leq b_1$, whence $b_1 = b_2$ and $f(b_1, b_2) = 0$. Therefore $f$ vanishes on $\Variety{\mathcal{F}_{\mathcal{T}}}$, and $f \in \GIdeal{\mathcal{F}_{\mathcal{T}}}$ by \cref{th:domain ideal nullstellensatz}, as required.

    Suppose now that $P_0 \cap P_1 = \emptyset$. We introduce two fresh distinct variables $z_0, z_1$ and we define the substitution
    \[
        \sigma(x) = \begin{cases} z_0, & x \in P_0,\\ z_1, & x \in P_1,\\ x, & x \notin P_0 \cup P_1.\end{cases}
    \]
    Let $\mathcal{C}_0$ be obtained from $\mathcal{C}$ by deleting every pinning constraint and applying $\sigma$ to each of the remaining ones, and set $\mathcal{T}_0 := (X_0, \BoolD, \mathcal{C}_0)$ with $X_0 := (X \setminus (P_0 \cup P_1)) \cup \{z_0, z_1\}$, so that both $z_0$ and $z_1$ occur in $X_0$ even when one of $P_0, P_1$ is empty. Every constraint of $\mathcal{C}_0$ has its relation in $\Gamma$, as the constraints with relation $\pin{0}$ or $\pin{1}$ are exactly the ones we deleted. We finally set
    \begin{equation}\label{eq:reduction-fix}
        \mathcal{T} := (X_0, \BoolD, \mathcal{C}_0 \cup \{G_a(z_0,z_1)\}), \qquad f := z_1 - z_0 ,
    \end{equation}
    which is an instance of $\IMP_1(\Gamma)$, since $G_a(z_0,z_1)$ is a constraint over $\Gamma$. As $\Gamma$ is fixed and finite, $R$, $a$ and $G_a$ are found by exhaustive search in time independent of $n$, thus $(f, \mathcal{T})$ is computed from $\mathcal{T}_{\Gpin}$ in polynomial time. We claim that
    \begin{equation}\label{eq:main-equivalence-fix}
        \mathcal{T}_{\Gpin} \text{ is satisfiable} \quad\Longleftrightarrow\quad f \notin \GIdeal{\mathcal{F}_{\mathcal{T}}} .
    \end{equation}

    Assume first that $\mathcal{T}_{\Gpin}$ has a solution $s = (s_x)_{x \in X}$. As $s$ satisfies every pinning constraint, $s_x = 0$ for $x \in P_0$ and $s_x = 1$ for $x \in P_1$. Define the tuple $b = (b_x)_{x \in X_0}$ by $b_{z_0} := 0$, $b_{z_1} := 1$ and $b_x := s_x$ for $x \in X \setminus (P_0 \cup P_1)$, hence $b_{\sigma(x)} = s_x$ for every $x \in X$. Each constraint of $\mathcal{C}_0$ is the image under $\sigma$ of a constraint satisfied by $s$, whence $b$ satisfies it; and $b$ satisfies $G_a(z_0,z_1)$ as well, since $b_{z_0} = 0 \leq 1 = b_{z_1}$. Hence $b \in \Variety{\mathcal{F}_{\mathcal{T}}}$ while $f(b) = 1 \neq 0$, and by the Nullstellensatz $f \notin \GIdeal{\mathcal{F}_{\mathcal{T}}}$.

    Assume conversely that $f \notin \GIdeal{\mathcal{F}_{\mathcal{T}}}$. Then there is a solution $b = (b_x)_{x \in X_0}$ of $\mathcal{T}$ with $f(b) \neq 0$, that is, with $b_{z_0} \neq b_{z_1}$. Since $b$ satisfies $G_a(z_0,z_1)$, \cref{lem:orientation} gives $b_{z_0} \leq b_{z_1}$, and therefore $b_{z_0} = 0$ and $b_{z_1} = 1$. Define the tuple $s = (s_x)_{x \in X}$ by $s_x := b_{\sigma(x)}$. Every non-pinning constraint of $\mathcal{C}$ is mapped by $\sigma$ to a constraint of $\mathcal{C}_0$, which $b$ satisfies, so $s$ satisfies it as well; and $s$ satisfies every pinning constraint too, since $s_x = b_{z_c} = c$ whenever $x \in P_c$. Hence $s$ is a solution of $\mathcal{T}_{\Gpin}$, which proves \eqref{eq:main-equivalence-fix}.

    Hence the map sending $\mathcal{T}_{\Gpin}$ to $(f, \mathcal{T})$ reduces $\CSP(\Gpin)$ to the complement of $\IMP_1(\Gamma)$, which is therefore \NP-hard, so that $\IMP_1(\Gamma)$ is \coNP-hard. Membership in \coNP{} follows from \cref{th:domain ideal nullstellensatz}, since a solution on which $f$ does not vanish certifies nonmembership and is checked in polynomial time.
\end{proof}

The threshold thus drops from $d \in \{0,1,2\}$ to $d \in \{0,1\}$, and we can finally state the complete dichotomy.

\begin{theorem}\label{th:IMP 0/1 hardness result for Boolean languages}
    Let $\Gamma$ be a finite Boolean constraint language. If the solution space of every relation in $\Gamma$ is closed under any of \{\textsc{Majority}, \textsc{Minority}, \textsc{Max}, \textsc{Min}\}, then $\IMP_d(\Gamma)$ can be solved in $n^{O(d)}$ time for any $d \geq 0$. Otherwise, $\IMP_1(\Gamma)$ is \coNP-complete. Furthermore, if $\CSP(\Gamma)$ is \NP-complete, then $\IMP_0(\Gamma)$ is \coNP-complete.
\end{theorem}

\subsection{The Boolean dichotomy for $\PC$-$\IMP_d(\Gamma)$}

We shift our focus from the general $\IMP_d(\Gamma)$ to the more constrained $\PC$-$\IMP_d(\Gamma)$. A $\PC$ certificate is verifiable in polynomial time and exists if and only if $f \in \langle \mathcal{F}\rangle$, so clocking a solver at its $n^{O(d)}$ bound and verifying its output decides $\IMP_d(\Gamma)$; thus solvability of $\PC$-$\IMP_d(\Gamma)$ implies $\IMP_d(\Gamma) \in \Ptime$. Contrapositively, since $\IMP_d(\Gamma) \in \coNP$ over any finite domain by the Strong Nullstellensatz, \coNP-hardness of $\IMP_d(\Gamma)$ obstructs solving $\PC$-$\IMP_d(\Gamma)$: it is then not solvable unless $\Ptime = \NP$.

For Boolean $\Gamma$, the degree-$0$ case $\PC$-$\IMP_0(\Gamma)$, the problem of refuting unsatisfiable $\CSP(\Gamma)$ instances in $\PC$, has a clean characterization.

\begin{theorem}[\cite{AtseriasO19}]\label{th:PC refutes bounded width}
    Let $\Gamma$ be a fixed constraint language. Then $\Gamma$ has bounded width if and only if $\PC$-$\IMP_0(\Gamma)$ can be solved with constant degree $\PC$ in polynomial time.
\end{theorem}

\noindent Among Boolean polymorphisms, each of $\{\textsc{Majority}, \textsc{Min}/\textsc{Max}, \mathbf{0}, \mathbf{1}\}$ gives bounded width, whereas a language whose only such polymorphism is \textsc{Minority} does not \cite{BartoKW17}.

This settles $d = 0$, but leaves the behaviour at $d \geq 1$ open for the languages closed under \textsc{Majority}, \textsc{Min}, or \textsc{Max}. Bortolotti et al.\ resolve these, in the broader setting of arbitrary finite domains $\Lambda$, through the dual-discriminator and semilattice polymorphisms --- which on the Boolean domain specialize to \textsc{Majority} and \textsc{Min}/\textsc{Max}.

\begin{theorem}[\cite{BortolottiMV25,BortolottiMV_ARXIV25}]\label{th:Bortolotti PC bounds for dual-discriminator and semilattice}
    Let $\Gamma$ be a finite constraint language over a domain $\Lambda$ closed under a semilattice or a dual-discriminator polymorphism. Then $\PC$-$\IMP_d(\Gamma)$ is solvable in polynomial time.
\end{theorem}

The separation between $\IMP$ and $\PC$-$\IMP$ is widest in the \textsc{Minority} case. Unsatisfiable affine systems over $\mathbb{F}_2$, which are \textsc{Minority}-closed, require $\PC$ refutations of degree $\Omega(n)$ over $\mathbb{Q}$ \cite{BussGIP01}, a bound \cite{AtseriasO19} recast in the algebraic $\CSP$ framework; so $\PC$-$\IMP_0(\Gamma)$ admits no constant-degree solution, while $\IMP_d(\Gamma) \in \Ptime$ \cite{BulatovRSTOC22,Mastrolilli21}. These languages thus separate $\IMP$ from $\PC$-$\IMP$ \emph{unconditionally}, sharpening the conditional obstruction mentioned earlier. They have no constant polymorphism, however: when $\Gamma$ has one, every instance is satisfiable and degree $0$ has nothing to refute. The separation then starts at degree $1$.

\begin{proposition}\label{prop:affine degree one}
    Let $\Gamma$ be a finite Boolean constraint language with $\textsc{Minority} \in \Pol(\Gamma)$ and $\textsc{Majority}, \textsc{Min}, \textsc{Max} \notin \Pol(\Gamma)$. Then $\PC$-$\IMP_1(\Gamma)$ is not solvable.
\end{proposition}

\begin{proof}
    Consider $R_{\oplus} := \{a \in \BoolD^6 \mid a_1 \oplus \cdots \oplus a_6 = 0\}$. Every affine operation is of the form $f(x_1, \ldots, x_m) = c \oplus x_1 \oplus \cdots \oplus x_m$ and, moreover, it preserves $R_{\oplus}$. Indeed, applying $f$ coordinatewise to $a^1, \ldots, a^m \in R_{\oplus}$ it returns a tuple whose six coordinates sum to $c \oplus \cdots \oplus c = 0$. The set $\Pol(\Gamma)$ is closed under composition, and Post classified such sets of Boolean operations \cite{BohlerCRV03,Post41}: one containing \textsc{Minority} but not \textsc{Majority} consists of affine operations. Every polymorphism of $\Gamma$ is thus affine and preserves $R_{\oplus}$, so that $\Pol(\Gamma) \subseteq \Pol(R_{\oplus})$. Hence $\Gamma$ pp-defines $R_{\oplus}$, and by \cref{th:pp-definability and PC reductions} it suffices to prove the claim for $\Gamma = \{R_{\oplus}\}$.

    Let $G = (V,E)$ be a $4$-regular graph of expansion $\epsilon$ on an odd number $n := |V|$ of vertices. The \emph{Tseitin system} $TS(G)$ has one variable $x_e$ per edge and the equation $\bigoplus_{e \ni v} x_e = 1$ per vertex. Summing all of them over $\mathbb{F}_2$ cancels every variable, so $TS(G)$ is unsatisfiable, as $n$ is odd; moreover no $\PC$ refutation of a polynomial encoding of $TS(G)$ has degree below $\epsilon n / 8$ \cite{BussGIP01}. Introduce two fresh variables $w_0, w_1$ and let $\mathcal{T}$ be the $\CSP(\{R_{\oplus}\})$ instance on $X := \{x_e\}_{e \in E} \cup \{w_0, w_1\}$ whose constraints are
    \begin{equation}\label{eq:affine tseitin instance}
        R_{\oplus}(x_{e_1}, x_{e_2}, x_{e_3}, x_{e_4}, w_0, w_1) \iff x_{e_1} \oplus x_{e_2} \oplus x_{e_3} \oplus x_{e_4} \oplus w_0 \oplus w_1 = 0, \qquad v \in V,
    \end{equation}
    where $e_1, \ldots, e_4 \in E$ are the edges adjacent to $v$. By summing all the $n$ equations on the right side, every solution of $\mathcal{T}$ satisfies $n(w_0 \oplus w_1) = 0$, and, because $n$ is odd, $w_0 = w_1$. The polynomial $f := w_1 - w_0$ therefore vanishes on every solution of $\mathcal{T}$, so that $f \in \GIdeal{\mathcal{F}_{\mathcal{T}}}$ by the Strong Nullstellensatz and $(f, \mathcal{T})$ is a promise instance of $\PC$-$\IMP_1(\{R_{\oplus}\})$.

    Let $\mathfrak{Q} = [q_1, \ldots, q_s = f]$ be a $\PC$ derivation of $f$ from $\mathcal{F}_{\mathcal{T}}$ of degree $\delta$, and let $r \colon \mathbb{Q}[X] \to \mathbb{Q}[\{x_e\}_{e \in E}]$ be the ring homomorphism fixing every $x_e$ and sending $w_1$ to $1$ and $w_0$ to $0$. The obtained sequence $[r(q_1), \ldots, r(q_s)]$ is a $\PC$ derivation from $r(\mathcal{F}_{\mathcal{T}})$ of degree at most $\delta$.

    Now $r(\mathcal{F}_{\mathcal{T}})$ is a polynomial encoding of $TS(G)$ as each constraint at $v$ becomes a polynomial vanishing exactly on the tuples with $\bigoplus_{e \ni v} x_e = 1$. Moreover, $r(q_s) = r(f) = 1$, so $[r(q_1), \ldots, r(q_s)]$ is a $\PC$ refutation of $TS(G)$ of degree at most $\delta$. The change of variables $y_e = 1 - 2 x_e$ transforms this into a refutation of an encoding of $TS(G)$ over $\{1,-1\}$, which by \cref{prop:encoding-independence} may be taken to be the system $TS_n(2)$ of \cite{BussGIP01}. Hence $\delta \geq \epsilon n / 8$ by \cite[Theorem 8]{BussGIP01}, which is the desired degree lower bound.
\end{proof}

Assembling these results yields the Boolean classification below.

\begin{theorem}\label{th:complete Boolean classification}
   Let $\Gamma$ be a fixed Boolean constraint language. If $\Gamma$ is closed under \textsc{Majority}, \textsc{Min}, or \textsc{Max}, then $\PC$-$\IMP_d(\Gamma)$ is solvable for every fixed $d \in \mathbb{N}$. If $\Gamma$ is closed under none of these but under \textsc{Minority}, then $\PC$-$\IMP_1(\Gamma)$ is not solvable, and already $\PC$-$\IMP_0(\Gamma)$ is not solvable when $\mathbf{0}, \mathbf{1} \notin \Pol(\Gamma)$; both bounds are unconditional. If $\Gamma$ is closed under none of \textsc{Majority}, \textsc{Min}, \textsc{Max}, \textsc{Minority} but under $\mathbf{0}$ or $\mathbf{1}$, then $\PC$-$\IMP_0(\Gamma)$ is solvable while $\PC$-$\IMP_1(\Gamma)$ is not solvable unless $\Ptime = \NP$. For every remaining $\Gamma$, $\PC$-$\IMP_0(\Gamma)$ is not solvable unless $\Ptime = \NP$.
\end{theorem}

This completes the Boolean picture. Over every domain the $d = 0$ case is settled by the bounded-width characterization of \cref{th:PC refutes bounded width}; for $d \ge 1$ over larger domains, however, little is known beyond the semilattice, dual-discriminator, and \textsc{Minority} cases, leaving the classification largely open.


\section{Algebras}\label{sect:algebras}

In this section we leverage the algebraic approach to $\CSP$s to reframe $\PC$-$\IMP$ reductions, such as the one in \cref{ex:semilattice}, in a much more elegant form. Throughout, $D$, $E$ denote the domains of algebras; when an algebra arises from a constraint language $\Gamma$ as $\textsf{Alg}(\Gamma)$ (see below), its domain is the language domain $\Lambda_\Gamma$ of \cref{sec:pp-interpretations} (so $\Lambda_\Gamma = D$, $\Lambda_\Delta = E$).

The algebraic notions we use are standard; see e.g.\ \cite{BulatovRSTOC22}. We recall the two operators linking relations and operations. For a constraint language $\Gamma$ over $D$, $\Pol(\Gamma)$ is the set of operations on $D$ preserving every relation in $\Gamma$ (its \emph{polymorphisms}); for a set of operations $\Psi$, $\Inv(\Psi)$ is the set of relations preserved by every $\psi \in \Psi$ (its \emph{invariants}).

\begin{definition}\label{def:algebras and tractability}
    Let $D$ be a finite domain and let $\Psi$ be a set of (multi-ary, in general) operations over $D$. The pair $(D, \Psi)$ is called an \emph{algebra}. For a constraint language $\Gamma$ over $D$ we write $\textsf{Alg}(\Gamma) := (D, \Pol(\Gamma))$ for the algebra of its polymorphisms.
    We say that $(D, \Psi)$ is $\PC$-tractable if $\PC$-$\IMP_d(\Gamma)$ is solvable for every finite constraint language $\Gamma \subseteq \Inv(\Psi)$.
\end{definition}

Cast as a property of algebras, $\PC$-tractability propagates from a tractable $\mathcal{D} = \textsf{Alg}(\Gamma)$ to the algebras built from it by the three standard constructions---subalgebras, finite direct powers, and homomorphic images: in each, the relations invariant under the new algebra are pp-definable or pp-interpretable in $\Gamma$, so \cref{th:pp-definability and PC reductions,th:pp-interpretations and PC reductions} transfer solvability.

\begin{definition}[Subalgebras and direct powers]
    Let $\mathcal{D} = (D, \Psi)$ be an algebra.
    \begin{itemize}
        \item \textbf{(Subalgebra)} \quad Let $E \subseteq D$ such that, for any $k$-ary operation $\psi \in \Psi$ and for any $b_1, \ldots, b_k \in E$, we have $\psi(b_1, \ldots, b_k) \in E$. In other words, $\psi$ is a polymorphism of $E$ or $E \in \textsf{Inv}(\Psi)$. The algebra $\mathcal{E} = (E, \Psi |_{E})$, where $\Psi |_{E}$ consists of the restrictions to $E$ of all operations in $\Psi$, is called \emph{subalgebra} of $\mathcal{D}$.
        \item \textbf{(Direct power)} \quad For a natural number $k$, the $k$-th \emph{direct power} $\mathcal{D}^k$ of $\mathcal{D}$ is the algebra $\mathcal{D}^k = (D^k, \Psi^k)$, where $\Psi^k$ consists of all the operations from $\Psi$ acting on $D^k$ component wise.
    \end{itemize}
\end{definition}

\begin{lemma}\label{lem:subalgebra tractable}
    Let $\mathcal{D} = (D, \Psi_D) = \textsf{Alg}(\Gamma_D)$ be a $\PC$-tractable algebra. If $\mathcal{E} = (E, \Psi_E)$ is a subalgebra of $\mathcal{D}$, then $\mathcal{E}$ is $\PC$-tractable as well.
\end{lemma}

\begin{proof}
    Let $\Gamma_E \subseteq \textsf{Inv}(\Psi |_E)$ be finite and consider the instance $(f, \mathcal{T}_E)$ of $\PC$-$\IMP_d(\Gamma_E)$, where $\mathcal{T}_E = (X, E, \mathcal{C}_E)$ is an instance of $\CSP(\Gamma_E)$. The definition of subalgebra implies $\Gamma_E \subseteq \textsf{Inv}(\Psi_D)$: $\Psi_D$ and $\Psi_E$ are the same operations, those in $\Psi_D$ acting on all of $D$ and those in $\Psi_E$ acting only on $E \subseteq D$.

    Next, we define a constraint language $\tilde{\Gamma}$ on the domain $D$ and an instance $\mathcal{T}_D$ of $\CSP(\tilde{\Gamma})$ equivalent to $\mathcal{T}_E$. Introduce the unary relation $R_E := E$, define $\tilde{\Gamma} = \Gamma_E \cup \{R_E\}$, the constraints $\mathcal{C}_D := \mathcal{C}_E \cup \{ (x_i, R_E)\}_{i \in [n]}$, and $\mathcal{T}_D := (X, D, \mathcal{C}_D)$. Then $\mathcal{T}_E$ and $\mathcal{T}_D$ have the same solution set.

    Let $\mathcal{F}_{\mathcal{T}_E}$ and $\mathcal{F}_{\mathcal{T}_D}$ be the generators from $\mathcal{T}_E$ and $\mathcal{T}_D$. Then $\langle \mathcal{F}_{\mathcal{T}_E} \rangle = \langle \mathcal{F}_{\mathcal{T}_D} \rangle$. Since $\Gamma_E \subseteq \textsf{Inv}(\Psi_D)$ and $R_E = E \in \textsf{Inv}(\Psi_D)$, we have $\tilde\Gamma \subseteq \textsf{Inv}(\Psi_D)$, so $\Gamma_D$ pp-defines $\tilde\Gamma$ and hence \cref{th:pp-definability and PC reductions} implies that there is a bounded-degree, polynomial-size $\PC$ derivation $\mathfrak{G} = [g_1, g_2, \ldots, f]$ of $f$ from $\mathcal{F}_{\mathcal{T}_D}$.

    Lastly, we simulate $\mathfrak{G}$ from $\mathcal{F}_{\mathcal{T}_E}$. Since
    \begin{align*}
        \mathcal{F}_{\mathcal{T}_D} = \{ f_C\}_{C \in \mathcal{C}_D} \cup \{f_D(x_i)\}_{i \in [n]},
    \end{align*}
    it suffices to observe:
    \begin{enumerate}[i.]
        \item For any constraint $C \in \mathcal{C}_E$ of arity $k$, write $f_C^{E}$ and $f_C^{D}$ for its interpolating polynomials over the grids $E^k$ and $D^k$, which differ in general. By \cref{th:domain ideal nullstellensatz} the ideal $\langle f_C^{E}, \{f_E(x)\}_{x \in Var(C)} \rangle$ is radical with variety $R_C$, on which $f_C^{D}$ vanishes, so $f_C^{D}$ lies in it; the local argument of \cref{prop:encoding-independence} derives it with constant degree and polynomial bit-size.
        \item For any constraint $S = (x_i, R_E) \in \mathcal{C}_D$ the polynomial $f_{S}(x_i) \in \langle f_E(x_i) \rangle$ by Hilbert Strong Nullstellensatz, hence is derivable efficiently by $\PC$.
        \item Any $f_D(x_i)$ is derivable from $f_E(x_i)$ via $f_D(x_i) = f_E(x_i) \cdot \Pi_{\rho \in D \setminus E}(x_i - \rho)$.
    \end{enumerate}
    Thus $\mathcal{F}_{\mathcal{T}_D}$, and therefore $f$, can be derived efficiently from $\mathcal{F}_{\mathcal{T}_E}$.
\end{proof}

\begin{lemma}\label{lem:direct power tractable}
    Let $\mathcal{D} = (D, \Psi) = \textsf{Alg}(\Gamma)$ be a $\PC$-tractable algebra. If $\mathcal{E} = \mathcal{D}^k$ is the $k$-th direct power of $\mathcal{D}$, then $\mathcal{E}$ is $\PC$-tractable as well.
\end{lemma}

\begin{proof}
    Let $\Delta \subseteq \textsf{Inv}(\Psi^k)$. We show that every $R \in \textsf{Inv}(\Psi^k)$ is pp-encoded in $\Gamma$, with $\ell = k$, $F = D^k$, and $\pi = \mathrm{id}_{D^k}$; that is, each variable over $D^k$ is encoded by $k$ variables over $D$, its coordinates.

    The encoded constraint is $\pi^{-1}(R)$: the $km$-ary relation over $D$ obtained by reading an $m$-ary $R \subseteq (D^k)^m$ coordinatewise (with $\pi = \mathrm{id}$, the same set of tuples, reinterpreted over $D$). The action of $\psi \in \Psi$ on $D^k$ is by definition its coordinatewise action, so $R \in \textsf{Inv}(\Psi^k)$ if and only if $\pi^{-1}(R) \in \textsf{Inv}(\Psi)$. Hence $\textsf{Pol}(\Gamma) \subseteq \textsf{Pol}(\pi^{-1}(R))$ and $\Gamma$ pp-defines $\pi^{-1}(R)$ \cite{BartoKW17}. The two remaining requirements are immediate: $F = D^k$ is the full $k$-ary relation, and $\pi^{-1}(=_{D^k})$ is coordinatewise equality $\bigwedge_{j \in [k]} (y_j = y_j')$. With $\pi = \mathrm{id}$ a bijection, this is a pp-encoding of $\Delta$ in $\Gamma$.
    
    Tractability follows from \cref{th:pp-encodings and PC reductions}: the resulting $\PC$-$\IMP$ instance is over $D$, with each $D^k$-variable replaced by its $k$ coordinates over $D$.
\end{proof}

\begin{definition}[Homomorphic image]\label{def:homomorphic image}
    Let $E$ be a set and $\chi : D \rightarrow E$ a surjective mapping such that for any $k$-ary $\psi \in \Psi$ and any $a_1, \ldots, a_k, b_1, \ldots, b_k \in D$, if $\chi(a_i) = \chi(b_i)$ for $i \in [k]$, then $\chi(\psi(a_1, \ldots, a_k)) = \chi(\psi(b_1, \ldots, b_k))$. The algebra $\mathcal{E} = (E, \Psi_\chi)$ is called a \emph{homomorphic image} of $\mathcal{D}$, where for every $\psi \in \Psi$ the set $\Psi_\chi$ contains $\psi/_\chi$ given by $\psi/_\chi(c_1, \ldots, c_k) = \chi(\psi(a_1, \ldots, a_k))$ for any $a_1, \ldots, a_k \in D$ with $c_i = \chi(a_i)$, $i \in [k]$. When $\chi$ is bijective, $\mathcal{E}$ is called an \emph{isomorphic image} of $\mathcal{D}$.
\end{definition}

\begin{lemma}\label{lem:homo pp-interpreted}
    Consider $\mathcal{D} = (D, \Psi) = \textsf{Alg}(\Gamma)$ and let $\mathcal{E} = (E, \Psi_\chi)$ be a homomorphic image of $\mathcal{D}$. Then every $R \in \textsf{Inv}(\Psi_\chi)$ is pp-interpreted in $\Gamma$ via $\ell = 1$, $F = D$, $\pi = \chi$. If $\chi$ is bijective, $R$ is pp-encoded.
\end{lemma}

\begin{proof}
    Set $\pi = \chi$, $\ell = 1$, $F = D$. We verify the three pp-definability requirements of pp-interpretation (\cref{def:pp-interpretability}).

    \emph{$F = D$.} The full unary relation, trivially pp-definable.

   \emph{$\pi^{-1}(=_E)$.} This is the relation $\theta_\chi := \{(a,b) \in D^2 : \chi(a) = \chi(b)\}$. Let $\psi \in \Psi$ be $k$-ary and let $(a_i, b_i) \in \theta_\chi$ for $i \in [k]$; writing $\bar a = (a_1, \dots, a_k)$, $\bar b = (b_1, \dots, b_k)$, \cref{def:homomorphic image} gives $\chi(\psi(\bar a)) = \chi(\psi(\bar b))$, so $(\psi(\bar a), \psi(\bar b)) \in \theta_\chi$ and $\theta_\chi \in \textsf{Inv}(\Psi)$.

    \emph{$\pi^{-1}(R)$, $R \in \textsf{Inv}(\Psi_\chi)$ of arity $k$.} By definition $\pi^{-1}(R) = \{\bar a \in D^k : (\chi(a_1), \ldots, \chi(a_k)) \in R\}$. Let $\psi \in \Psi$ be $h$-ary, take $\bar a_1, \ldots, \bar a_h \in \pi^{-1}(R)$, and put $\bar a = \psi(\bar a_1, \ldots, \bar a_h)$ componentwise, so $a_i = \psi(a_{i,1}, \ldots, a_{i,h})$. Then
    \[
        \chi(a_i) = \chi(\psi(a_{i,1}, \ldots, a_{i,h})) = (\psi/_\chi)(\chi(a_{i,1}), \ldots, \chi(a_{i,h})).
    \]
    Since $\psi/_\chi \in \Psi_\chi$ and $R \in \textsf{Inv}(\Psi_\chi)$, applying $\psi/_\chi$ to the $h$ tuples $(\chi(a_{1,j}), \ldots, \chi(a_{k,j})) \in R$, $j \in [h]$, yields a tuple of $R$, namely $(\chi(a_1), \ldots, \chi(a_k))$; hence $\bar a \in \pi^{-1}(R)$ and $\textsf{Pol}(\Gamma) \subseteq \textsf{Pol}(\pi^{-1}(R))$.

    Since $\Gamma$ pp-defines every relation preserved by $\Pol(\Gamma)$ \cite{BartoKW17}, it pp-defines $F$, $\theta_\chi$, and each $\pi^{-1}(R)$; that is, $\Gamma$ pp-interprets every $R \in \textsf{Inv}(\Psi_\chi)$. If $\chi$ is bijective, so is $\pi$, and the interpretation is a pp-encoding.
\end{proof}

\begin{proposition}\label{prop:homo reduction}
    If $\mathcal{E}$ is a homomorphic image of $\mathcal{D}$, then $\PC$-$\IMP$ from $\mathcal{E}$ $\PC$-reduces to $\PC$-$\IMP$ from $\mathcal{D}$; in particular $\mathcal{E}$ is $\PC$-tractable whenever $\mathcal{D}$ is.
\end{proposition}

\begin{proof}
    By \cref{lem:homo pp-interpreted}, every language of $\mathcal{E}$ is pp-interpreted in $\Gamma$, so the reduction follows from \cref{th:pp-interpretations and PC reductions}. When $\chi$ is bijective the interpretation is a pp-encoding.
\end{proof}

Together, \cref{lem:subalgebra tractable,lem:direct power tractable,prop:homo reduction} show that $\PC$-tractability is preserved under subalgebras, finite direct powers, and homomorphic images; it therefore propagates to every algebra obtained from a $\PC$-tractable $\mathcal{D}$ by these operations.

\subsection{Median operations and fixed-value majorities}\label{sect:majorities}
We apply the transfer principle to obtain new tractable classes over ternary and larger domains. The engine is the \emph{median} operation on a finite chain, which reduces to the Boolean majority algebra; its three-element case yields the \emph{fixed-value majorities} over $\{0,1,2\}$, advancing the ternary $\IMP_d$ classification.

\begin{proposition}\label{prop:median-chain}
    For every $\ell \geq 1$, let $m_\ell$ be the ternary operation on $\{0,1,\ldots,\ell\}$ returning the middle one of its three arguments, i.e. $m_\ell(a,b,c) = a + b + c - \max(a,b,c) - \min(a,b,c)$. Then the algebra $(\{0,1,\ldots,\ell\}, m_\ell)$ is $\PC$-tractable.
\end{proposition}
\begin{proof}
    Let $\mathcal{B} = (\{0,1\}, \textsc{Maj})$ with $\textsc{Maj}$ the ternary Boolean majority. Consider $F = \{(a_1, \ldots, a_{\ell}) \in \{0,1\}^{\ell} : a_1 \geq a_2 \geq \ldots \geq a_{\ell}\}$ and note that it is closed under componentwise $\textsc{Maj}$, implying that $(F, \textsc{Maj}^{\ell})$ is a subalgebra of $\mathcal{B}^{\ell}$. We recall that $\mathcal{B}$ is $\PC$-tractable \cite{BortolottiMV_ARXIV25}, therefore $(F, \textsc{Maj}^{\ell})$ is also $\PC$-tractable. The map $\chi : F \to \{0, \ldots, \ell\}$ with $\chi(a_1, \ldots, a_{\ell}) = \sum_{j = 1}^{\ell} a_j$ is a bijection, and it is monotone for the componentwise order on $F$: if $x \leq y$ componentwise then $\chi(x) \leq \chi(y)$. Any $b_1, b_2, b_3 \in F$ are pairwise comparable, and suppose $b_1 \leq b_2 \leq b_3$. Then $\textsc{Maj}^{\ell}(b_1, b_2, b_3) = b_2$, while $\chi(b_1) \leq \chi(b_2) \leq \chi(b_3)$ gives $m_\ell(\chi(b_1), \chi(b_2), \chi(b_3)) = \chi(b_2)$. Hence $\chi(\textsc{Maj}^{\ell}(b_1, b_2, b_3)) = m_\ell(\chi(b_1), \chi(b_2), \chi(b_3))$, and $(\{0, \ldots, \ell\}, m_\ell)$ is an isomorphic image of $(F, \textsc{Maj}^{\ell})$, concluding the proof.
\end{proof}

Specializing to $\ell = 2$ gives the fixed-value majorities. For $i \in \{0,1,2\}$, the \emph{fixed-value majority} $\mu_i$ over $\{0,1,2\}$ is the majority returning the constant $i$ on all-distinct triples: $\mu_i(a,b,c) = \textsc{Maj}(a,b,c)$ if $|\{a,b,c\}| \leq 2$, and $\mu_i(a,b,c) = i$ otherwise.

\begin{corollary}\label{cor:012-majorities}
    For every $i \in \{0,1,2\}$, the algebra $(\{0,1,2\}, \mu_i)$ is $\PC$-tractable.
\end{corollary}
\begin{proof}
    For $\ell = 2$ we have $m_2 = \mu_1$, so $(\{0,1,2\}, \mu_1)$ is $\PC$-tractable by \cref{prop:median-chain}. For $i \in \{0,2\}$, the algebra $(\{0,1,2\}, \mu_i)$ is isomorphic to $(\{0,1,2\}, \mu_1)$ via any bijection of $\{0,1,2\}$ sending $i$ to $1$.
\end{proof}

\cref{cor:012-majorities} illustrates the versatility of these reductions: the $\PC$-tractability of $(\{0,1,2\},\mu_i)$ follows with no further certificate construction, purely from the isomorphism built in \cref{prop:median-chain}.
Exhibiting that isomorphism---and, more generally, a homomorphic image---is the only work; transfer is then automatic.

\section{Future work}\label{sect:future}

The framework of \cref{sect:algebras} makes $\PC$-tractability a property closed under subalgebras, finite powers, and homomorphic images, so that the tractability of an entire class of languages follows from that of a few representative algebras generating it. The semilattice (\cref{ex:semilattice}) and the median over a chain (\cref{prop:median-chain}), both reducing to a $\PC$-tractable Boolean base, are two such representatives.
The natural program is to identify the others: to find a set of base algebras from which, under the reductions of this paper, every $\PC$-tractable language is generated --- the proof-complexity counterpart of the representative operations that organize the $\CSP$ and $\IMP$ classifications.

Several cases remain open. Over $\{0,1,2\}$, the majorities outside the fixed-value family are not classified even for $\IMP_d$, and so neither for $\PC$-$\IMP_d$ (\cref{sect:majorities}); over larger domains, no classification of $\PC$-$\IMP_d$ is known beyond the semilattice, dual-discriminator, and \textsc{Minority} cases (\cref{sect:csp-pc-imp}).
There, for bounded-width languages, the only general lower bound is the conditional hardness inherited from $\IMP$---$\coNP$-hardness of $\IMP_d(\Gamma)$ rules out solving $\PC$-$\IMP_d(\Gamma)$ unless $\Ptime=\NP$; obtaining unconditional $\PC$ lower bounds outside the affine case is what a complete picture would require.

\newpage

\bibliography{main}

\end{document}